\documentclass[11pt]{article}
\usepackage{authblk}
\usepackage{amssymb}
\usepackage{amsthm}
\usepackage{mathtools}
\usepackage{amsfonts}
\usepackage{amsmath}
\usepackage{graphicx} 
\usepackage{amsmath,amsthm,amssymb}
\usepackage[margin=1in]{geometry}
\usepackage[colorlinks=true]{hyperref}
\usepackage{thmtools, thm-restate}
\usepackage{algorithm,enumitem}
\usepackage{algorithm}
\usepackage{algpseudocode}

\declaretheorem{theorem}
\usepackage{float}

\newtheorem{lemma}{Lemma} 
 
\newtheorem{corollary}{Corollary}
\newtheorem{definition}{Definition}

\newtheorem{proposition}{Proposition}

\newtheorem*{definition*}{Definition}

\newcommand{\occ}{\operatorname{occ}}
\newcommand{\occin}{\operatorname{occ}^{\mathrm{in}}}
\newcommand{\occal}{\operatorname{occ}^{\mathrm{al}}}

\newcommand{\Unif}{\mathrm{Unif}}
\newcommand{\Prob}{\mathbb{P}}
\newcommand{\Exp}{\mathbb{E}}

\newcommand{\Geom}{\mathrm{Geom}}

\newcommand{\N}{\mathbb{N}}

\newcommand{\abs}[1]{\left|#1\right|}

\newcommand{\poly}{\mathrm{poly}}
\newcommand{\A}{\Sigma}
\newcommand{\floor}[1]{\left\lfloor #1\right\rfloor}
\newcommand{\alocc}{\operatorname{occ}^{\mathrm{al}}}
\newcommand{\E}{\mathbb{E}}

\usepackage{tikz}
\usetikzlibrary{automata, positioning}

\usepackage{graphicx}

\title{Efficient Constructions of Finite-State Independent Normal Pairs}

\author{Subin Pulari}
\affil[1]{
	National Research University Higher School of Economics Moscow
}

\date {\today}

\begin{document}
	
	\maketitle
\begin{abstract}
	Finite-state independence is a robust notion of algorithmic independence for infinite words.
	It was introduced for general infinite words by Becher, Carton, and Heiber via deterministic asynchronous two-tape finite automata.
	Álvarez, Becher, and Carton then studied the normal case and characterized finite-state independence in terms of deterministic finite-state \emph{shufflers}.
	A shuffler is a finite automaton that reads from two input tapes $x,y\in\Sigma^\infty$ and, at each step, chooses one tape to read next, outputs the symbol read,
	and updates its state based only on that output symbol.
	In terms of this characterization, two normal sources are finite-state independent if \emph{every} deterministic finite-state way of shuffling (interleaving) them still produces a normal sequence.
	Álvarez, Becher, and Carton also exhibited a doubly-exponential procedure to construct a finite-state independent normal pair and posed the following questions:
	(1) can one compute finite-state independent normal pairs efficiently, improving their doubly-exponential procedure; and
	(2) given a normal word $x$, can one effectively construct a normal word $y$ that is finite-state independent from $x$?
	
	We answer both questions by explicit deterministic constructions.
	First, we give a deterministic polynomial-time algorithm that, on input $N$, outputs the first $N$
	symbols of two normal words $x$ and $y$ such that for every shuffler $S$, the shuffled output
	$S(x,y)$ is normal; hence $(x,y)$ is finite-state independent.
	The construction is driven by a computable potential that aggregates conditional failure
	probabilities of a sparse family of normality tests for the first $N$ shufflers, together with
	uniform tail bounds that remain valid under conditioning on short fixed prefixes.
	
	Second, we solve the one-sided companion problem effectively.
	Given any computable normal word $x\in\Sigma^\infty$, we give an explicit deterministic construction of a computable normal word
	$y\in\Sigma^\infty$ such that for every shuffler $S$, the shuffled output $S(x,y)$ is normal.
	In particular, $x$ and $y$ are finite-state independent by the shuffler characterization theorem.
\end{abstract}

\section{Introduction}

Finite-state independence is a notion of algorithmic independence for infinite words in which the
observer is restricted to finite memory.  It was introduced by Becher, Carton, and Heiber~\cite{Becher2018}
via \emph{conditional finite-state compression}: one measures how well a word $x\in\Sigma^\infty$ can
be compressed by a one-to-one finite-state transducer when a second word $y\in\Sigma^\infty$ is
available as an oracle input.  Technically, this oracle-assisted compression is modeled
by deterministic asynchronous transducers with two input tapes and one output tape, formalized as
\emph{$2$-deterministic $3$-automata}~\cite{Becher2018}.  Two words are finite-state independent if neither
helps to compress the other in this model.  Under the uniform product measure on
$\Sigma^\infty\times\Sigma^\infty$, almost every pair is finite-state independent~\cite{Becher2018}.

Normality is the property that an infinite digit sequence has the expected uniform frequencies for all fixed finite blocks.  A word $z\in\Sigma^\infty$ is \emph{normal} (in the sense of
Borel~\cite{Borel1909}) if, for each block length $m$, every $m$-symbol pattern appears with the
limiting frequency predicted by the uniform distribution.  Borel proved that almost every infinite
word is normal~\cite{Borel1909}.  It is conjectured that the digit expansions of familiar constants
such as $\sqrt{2}$, $\pi$, and $e$ are normal, but this is unknown.
There are nevertheless explicit constructions of normal sequences; the classical example is
Champernowne's number~\cite{Champernowne1933}, whose base-$10$ digit sequence is obtained by
concatenating the successive integers,
\[
C_{10}\ :=\ 0.12345678910111213141516\cdots .
\]
(See Kuipers--Niederreiter~\cite{KuipersNiederreiterUniform} for background on normal numbers and equivalent formulations.)

Álvarez, Becher, and Carton~\cite{ABC} investigated finite-state independence in the normal setting
and gave several equivalent characterizations using deterministic asynchronous two-tape automata.
One characterization is particularly convenient for constructions: when $x$ and $y$ are normal over
the same alphabet, they are finite-state independent if and only if \emph{every} deterministic
finite-state \emph{shuffler} preserves normality, i.e., $S(x,y)$ is normal for every shuffler $S$.
A shuffler is a deterministic finite automaton that reads from two input tapes and outputs an
interleaving of the two inputs, choosing at each step which tape to read next and updating its
state based only on the symbol just output~\cite{ABC}.

Based on this shuffler characterization, Álvarez, Becher, and Carton gave an explicit algorithm that
constructs a finite-state independent normal pair, but its running time is doubly exponential in the
output length~\cite{ABC}.  They point out that the same type of bottleneck already appears in
Turing's construction of absolutely normal numbers~\cite{Turing1992}: one enforces an expanding family
of frequency constraints by repeatedly intersecting ``good'' sets at rapidly growing cutoff lengths,
and the cost is dominated by checking these large constraint families~\cite{ABC}.  Motivated by this,
they asked two concrete questions~\cite{ABC}:
\begin{enumerate}
	\item Can one compute finite-state independent normal pairs efficiently, improving the doubly-exponential
	procedure?
	\item Given a normal word $x$, can one effectively construct a normal word $y$ that is finite-state
	independent from $x$ (in particular, for explicit examples such as Champernowne's word)?
\end{enumerate}
These questions should be contrasted with the single-source setting: there are highly efficient
algorithms for producing (absolutely) normal numbers/words~\cite{Lutz2021}, but finite-state
independence demands uniform normality preservation under \emph{all} finite-state shuffles.

\paragraph{Our contributions.}
We answer both questions by explicit deterministic constructions.

\smallskip
\noindent\textbf{(1) Efficient construction of an independent normal pair.}
We give a deterministic polynomial-time algorithm that, on input $N$, outputs the first $N$ symbols
of two normal words $x$ and $y$ such that $S(x,y)$ is normal for every shuffler $S$; hence $(x,y)$ is
finite-state independent by the characterization theorem~\cite{ABC}.
At a high level, our construction replaces nested-set enforcement by a \emph{local} greedy extension
strategy controlled by a computable potential: at each step we extend one symbol of one tape so as to
keep an upper bound on the \emph{total conditional failure mass} of a sparse family of shuffler normality
tests.  The main technical inputs are (a) uniform tail bounds for aligned block-frequency deviations that
remain valid under conditioning on short fixed prefixes, and (b) a dynamic program that computes the
relevant conditional probabilities in polynomial time.

\smallskip
\noindent\textbf{(2) Effective one-sided construction: a computable companion for a given computable normal input.}
Given any computable normal word $x\in\Sigma^\infty$, we construct a computable normal word
$y\in\Sigma^\infty$ such that $S(x,y)$ is normal for every shuffler $S$.
In particular, this answers the question raised in~\cite{ABC} about constructing a finite-state independent companion for Champernowne's normal word: since Champernowne's word is computable and normal~\cite{Champernowne1933}, our procedure produces a computable normal $y$ that is finite-state independent from it.
The proof has two stages.  First we prove an almost-sure statement: for any fixed shuffler $S$, if $Y$
is i.i.d.\ uniform then $S(x,Y)$ is normal with probability $1$.
Second, we effectivize this by defining, for each finite family of shufflers and tests, a clopen set of
high measure of ``good'' $y$'s, and then computing a single computable $y$ in the intersection of all
these sets via a standard effective compactness/extraction argument.

\paragraph{Organization.}
Section~\ref{sec:preliminaries} introduces the main definitions and notation used throughout, including finite-state independence, shufflers, and normality.
Section~\ref{sec:ptime-construction} gives the polynomial-time construction of a finite-state independent normal pair.
Section~\ref{sec:computable-one-sided} gives the effective one-sided construction for a fixed computable normal input.
	
\section{Preliminaries}
\label{sec:preliminaries}

Fix a finite alphabet $\Sigma$ of size $k\ge 2$.
We write $\Sigma^n$ for the set of words of length exactly $n$, $\Sigma^{\le n}:=\bigcup_{j=0}^n \Sigma^j$,
$\Sigma^*:=\bigcup_{n\ge 0}\Sigma^n$ for the set of finite words, and $\Sigma^\infty$ for the set of infinite words.
The empty word is denoted by $\lambda$, and $|u|$ denotes the length of a finite word $u$.
For $x\in\Sigma^\infty$ and $n\in\N$, we write $x[1..n]\in\Sigma^n$ for the length-$n$ prefix of $x$.

We equip $\Sigma^\infty$ with the usual product (Cantor) topology. For $u\in\Sigma^*$, the cylinder determined by $u$ is
\[
[u]\ :=\ \{x\in\Sigma^\infty:\ x[1..|u|]=u\}.
\]
Cylinder sets are clopen. Moreover, since $\Sigma^\infty$ is compact and cylinders form a basis of the topology,
every clopen set is a finite union of cylinders. Let $\nu$ be the uniform distribution on $\Sigma$, i.e.\ $\nu(a)=1/k$ for all $a\in\Sigma$.
We equip $\Sigma^\infty$ with the product measure $\nu^{\otimes \N}$.
On $\Sigma^\infty\times\Sigma^\infty$ we use the product measure $\mu:=\nu^{\otimes \N}\otimes \nu^{\otimes \N}$.
For $u,v\in\Sigma^*$ we write $[u]\times[v]$ for the corresponding cylinder rectangle, and we write
$\mu(\,\cdot\mid [u]\times[v])$ for conditional probability with respect to this rectangle
(which is well-defined since $\mu([u]\times[v])=k^{-(|u|+|v|)}>0$).

We use both overlapping and aligned block-occurrence counts.
For a finite word $u\in\Sigma^*$ and a block $w\in\Sigma^m$, define the (overlapping) occurrence count
\[
\occin(u,w)\ :=\ \#\{\,1\le t\le |u|-m+1:\ u[t..t+m-1]=w\,\}.
\]
Throughout the paper we write $\occ(u,w)$ for $\occin(u,w)$.
We also use aligned (non-overlapping) counts at a fixed block length: for $u\in\Sigma^*$ and $w\in\Sigma^m$, define
\[
\occal(u,w)\ :=\ \#\Bigl\{\,0\le j<\Bigl\lfloor\frac{|u|}{m}\Bigr\rfloor:\ u[jm+1\,..\,jm+m]=w\,\Bigr\}.
\]
More generally, for an alignment length $r\ge 1$, a block $w\in\Sigma^r$, and $u\in\Sigma^*$, writing
$M=\bigl\lfloor |u|/r\bigr\rfloor$, define
\[
\alocc_{w,r}(u)\ :=\ \#\{\,0\le j < M:\ u[jr+1\,..\,jr+r]=w\,\}.
\]

Now we define normality in terms of overlapping block-occurrence counts.

\begin{definition}[Normality {\cite{Borel1909}}]
	A word $z\in\Sigma^\infty$ is \emph{(Borel) normal} if for every $m\ge 1$ and every $w\in\Sigma^m$,
	\[
	\lim_{n\to\infty}\frac{\occ(z[1..n],w)}{n}=k^{-m}.
	\]
\end{definition}
Normality can equivalently be defined using aligned block counts $\occal(\cdot,\cdot)$ (with normalization by
$\lfloor n/m\rfloor$); see \cite{KuipersNiederreiterUniform}. We will also use the standard computability notion for infinite words.

\begin{definition}[Computable infinite word]
	A word $z\in\Sigma^\infty$ is \emph{computable} if there exists an algorithm that, on input $n\in\N$ (in binary),
	outputs the prefix $z[1..n]$.
\end{definition}

Finite-state independence for arbitrary infinite words was introduced by Becher, Carton, and Heiber~\cite{Becher2018}
via conditional finite-state compression, using deterministic asynchronous transducer models (formalized as
$2$-deterministic $3$-automata).  In the normal setting, Álvarez, Becher, and Carton~\cite{ABC} gave a characterization
in terms of deterministic shufflers, which is the formulation used throughout this paper.

\begin{definition}[Deterministic shuffler {\cite[Def.~12]{ABC}}]
	\label{def:shuffler}
	A (deterministic) \emph{shuffler} is a tuple $S=(Q,q_0,\delta,\tau)$ where $Q$ is a finite set of states,
	$q_0\in Q$ is a start state, $\delta:Q\times\Sigma\to Q$ is a transition function, and $\tau:Q\to\{1,2\}$ is a tape-choice
	function.  Given $(x,y)\in\Sigma^\infty\times\Sigma^\infty$, the output $z=S(x,y)\in\Sigma^\infty$ is produced by the
	following infinite process: maintain head positions $a=b=0$ and a current state $q$; at each step let $\ell=\tau(q)$, read
	the next unused symbol from tape $\ell$ (from $x$ if $\ell=1$, from $y$ if $\ell=2$), output that symbol $\alpha$, and
	update $q\leftarrow \delta(q,\alpha)$.
\end{definition}
In \cite{ABC}, shufflers are presented as $2$-deterministic $3$-automata whose transitions are of one of two copy-types
(copy from tape~1 or copy from tape~2).  Determinism implies that all transitions leaving a fixed state have the same
copy-type, hence induce a well-defined tape-choice function $\tau$, yielding Definition~\ref{def:shuffler}.

For normal words over a common alphabet, we will use the shuffler-based formulation as our working definition of
finite-state independence.

\begin{definition}[Finite-state independence for normal words]
	Let $x,y\in\Sigma^\infty$ be normal. We say that $x$ and $y$ are \emph{finite-state independent} if for every shuffler $S$
	over $\Sigma$, the shuffled output $S(x,y)$ is normal.
\end{definition}

The next theorem, proved by Álvarez, Becher, and Carton~\cite{ABC}, shows that this formulation is equivalent to the
original definition of finite-state independence introduced by Becher, Carton, and Heiber~\cite{Becher2018}.

\begin{theorem}[Shuffler characterization {\cite[Thm.~1]{ABC}}]
	\label{thm:abc}
	Let $x,y\in\Sigma^\infty$ be normal. Then $x$ and $y$ are finite-state independent in the sense of
	Becher, Carton, and Heiber~\cite{Becher2018} if and only if for every shuffler $S$ over $\Sigma$,
	the output $S(x,y)$ is normal.
\end{theorem}

	\section{Polynomial-time construction of a finite-state independent normal pair}
	\label{sec:ptime-construction}
	
	In \cite{ABC}, \'{A}lvarez et al.\ asked whether one can compute, in polynomial time, a finite-state independent pair of normal sequences $x,y\in\Sigma^\infty$ (Question~3 from Section~6 of \cite{ABC}). In the following theorem,
	we give a positive answer.
	
	\begin{theorem}
		\label{thm:main}
		There is a deterministic algorithm that, on input $N$, outputs prefixes $(x{\upharpoonright}N,\;y{\upharpoonright}N)$
		in time $N^{O(1)}$ such that the infinite limit sequences $x,y\in\Sigma^\infty$ satisfy:
		\begin{enumerate}
			\item $x$ and $y$ are normal;
			\item for every shuffler $S$, the shuffled output $S(x,y)$ is normal.
		\end{enumerate}
		In particular, $(x,y)$ is a finite-state independent pair by the characterization theorem of \cite{ABC}.
	\end{theorem}
	
	\paragraph{Proof organization.}
	\label{par:proof-organization}
	We build $x$ and $y$ incrementally, one symbol at a time, by maintaining a potential that upper-bounds the
	(total) conditional probability of failing any currently relevant normality constraint.
	In Subsection~\ref{subsec:constraints} we define, for each output length $n$, a finite family $F_n$ of aligned
	block-frequency constraints for the first $t_n$ shufflers.
	In Subsection~\ref{subsec:conditional} we prove a uniform tail bound showing that each individual constraint has
	tiny failure probability even after conditioning on short fixed prefixes.
	In Subsection~\ref{subsec:potential} we aggregate these conditional failure probabilities into a bad-mass functional
	$B_n(u,v)$ and then into a rolling potential $\Phi_L(u,v)$ over a sparse checkpoint sequence $N_j$.
	Subsection~\ref{subsec:algorithm-correctness} gives the greedy algorithm (Algorithm~\ref{alg:compute}) and proves
	correctness via an integrality argument at checkpoints (Proposition~\ref{prop:checkpoints}) and an interpolation
	step that yields normality for every shuffler output.
	Finally, Subsection~\ref{subsec:ptime-impl} shows that the algorithm runs in time $N^{O(1)}$ by a dynamic program
	for computing each conditional probability term.

	\subsection{Aligned block counts and the $F_n$ constraints}
	\label{subsec:constraints}
	
	We now set up the finite families of frequency constraints that our construction will enforce at a sparse
	sequence of output lengths.  These constraints are stated in terms of aligned block counts and are imposed
	simultaneously for an initial segment of the shuffler enumeration. Fix a concrete enumeration $(S_i)_{i\ge 1}$ of all shufflers, chosen so that the map $i\mapsto S_i$ is efficiently
	computable and the size of $S_i$ is polynomially bounded in $i$.
	To this end, fix a canonical explicit binary encoding $\mathrm{enc}(S)$ of a shuffler $S=(Q,q_0,\delta,\tau)$ in which the
	state set is identified with $\{1,\dots,\abs{Q}\}$ and the transition table $\delta:Q\times\A\to Q$ together with the
	tape-choice function $\tau:Q\to\{1,2\}$ are written out in full.  For fixed alphabet size $k=\abs{\A}$, such an encoding
	has length $\abs{\mathrm{enc}(S)}=\Theta(\abs{Q}\,k\log\abs{Q})$.
	
	Let $(b_i)_{i\ge 1}$ be the length-lexicographic enumeration of all binary strings.
	Define $S_i$ to be the shuffler obtained by decoding $b_i$ if $b_i$ is a valid encoding of some shuffler, and otherwise
	let $S_i$ be a fixed trivial shuffler (e.g.\ the one that always reads tape~1).
	With this convention, given $i$ one can compute $b_i$ and check validity/perform decoding in time
	$\poly(\abs{b_i})=\poly(\log i)$, so the enumeration is effective with explicit resource bounds.
	Moreover, since the encoding is explicit, any valid $b_i$ can describe only shufflers whose number of states satisfies
	$\abs{Q_i}\le \poly(\abs{b_i})=\poly(\log i)\le \poly(i)$, a bound that will be used in the runtime analysis.

	Next we specify, for each output length $n$, the range of shufflers, block lengths, and error tolerance that will
	constitute the finite constraint family $F_n$.
	
	For $n\ge 3$ define:
	\[
	t_n := n,\qquad
	\ell_n := \left\lfloor \tfrac{1}{3}\log_k n\right\rfloor,\qquad
	\varepsilon_n := 2\sqrt{\frac{\log n\cdot \log_k n}{n}}.
	\]
	For each $r\le \ell_n$, set $m_{n,r} := \floor{n/r}$.
	
	Using these parameters, we now formalize what it means for one specific shuffler output to have the correct
	aligned frequency for one specific word.
	
	\begin{definition}
		\label{def:good-event}
		For a shuffler $S$, $n\ge 3$, $1\le r\le \ell_n$, and $w\in\A^r$, define
		\[
		E_{S}(n,r,w)\ :=\ \Bigl\{(x,y):\ \bigl|\alocc_{w,r}(S(x,y){\upharpoonright}n) - m_{n,r}/k^r \bigr|
		\ <\ \varepsilon_n\, m_{n,r}\Bigr\}.
		\]
	\end{definition}
	
	Finally, $F_n$ is the intersection of all such good events over the first $t_n$ shufflers, all block lengths up to
	$\ell_n$, and all words of the corresponding length.
	
	\begin{definition}[$F_n$]
		\label{def:Fn}
		Define
		\[
		F_n\ :=\ \bigcap_{i=1}^{t_n}\ \bigcap_{r=1}^{\ell_n}\ \bigcap_{w\in\A^r} E_{S_i}(n,r,w).
		\]
	\end{definition}
	
	\subsection{Conditional tail bounds under short prefix conditioning}
	\label{subsec:conditional}
	
	We will repeatedly use a standard multiplicative Chernoff bound; we record a convenient form.
	
	\begin{lemma}
		\label{lem:chernoff}
		Let $X\sim \mathrm{Bin}(M,p)$ and assume $0<\delta\le 1$.
		Then
		\[
		\Pr\bigl[\ |X-Mp|\ge \delta Mp\ \bigr]\ \le\ 2\exp\!\left(-\frac{\delta^2}{3}\,Mp\right).
		\]
	\end{lemma}
	
	The next lemma is the key uniform estimate: it bounds the conditional failure probability of a \emph{single}
	constraint even after fixing short prefixes on both input tapes.
	
	\begin{lemma}
		\label{lem:cond-one}
		Fix $n\ge 3$, $1\le r\le \ell_n$, and $w\in\A^r$.
		Let $S$ be any shuffler. Let $L\ge 0$ and $u,v\in\A^L$.
		Write $m=m_{n,r}=\floor{n/r}$.
		Assume $4L \le \varepsilon_n\, m$ and $\varepsilon_n \le k^{-r}$.
		Then
		\[
		\mu\bigl(E_S(n,r,w)^c\mid [u]\times[v]\bigr)\ \le\
		2\exp\!\left(-\frac{\varepsilon_n^2}{12}\,m\,k^r\right).
		\]
	\end{lemma}
	
	\begin{proof}
		Work under the conditional distribution $\mu(\cdot\mid [u]\times[v])$ and write $Z:=S(X,Y){\upharpoonright}n$.
		Partition $Z$ into the $m=\lfloor n/r\rfloor$ aligned length-$r$ blocks used by $\alocc_{w,r}$.
		
		During the production of the first $n$ output symbols, every time the shuffler reads an input symbol from tape~1
		at a position $\le L$ it outputs the corresponding fixed letter of $u$, and similarly for tape~2 with $v$.
		Let $T\subseteq\{1,2,\dots,n\}$ be the (random) set of output times at which the read input position lies in a fixed
		prefix (i.e.\ at which the output letter is one of the $2L$ fixed prefix symbols).
		Since each of the first $L$ symbols on each tape can be read at most once, we have the deterministic bound
		\[
		\abs{T}\ \le\ 2L.
		\]
		
		Now construct an auxiliary length-$n$ word $U\in\A^n$ as follows.
		For each $t\in\{1,\dots,n\}$, set
		\[
		U[t]\ :=\
		\begin{cases}
			Z[t], & t\notin T,\\
			R_t, & t\in T,
		\end{cases}
		\]
		where $(R_t)_{t\in T}$ are fresh independent uniform letters from $\A$, independent of $(X,Y)$.
		
		\smallskip
		\noindent\textbf{Claim 1:} $U$ is a uniformly random word in $\A^n$ (equivalently, its letters are i.i.d.\ uniform).
		
		\smallskip
		\noindent
		Let $\mathcal{F}_{t-1}$ be the $\sigma$-field generated by all symbols of $X$ and $Y$ that have been read by $S$
		during the first $t-1$ output steps (together with all auxiliary letters $R_{t'}$ for $t'<t$).
		At time $t$, the shuffler's state, tape choice $\tau(q)$, and head positions are $\mathcal{F}_{t-1}$-measurable,
		since they are deterministic functions of the previously read input symbols.
		If $t\in T$, then $U[t]=R_t$ is uniform on $\A$ and independent of $\mathcal{F}_{t-1}$ by construction.
		If $t\notin T$, then the shuffler reads the next unused symbol from one tape at some position $j>L$.
		Under the product measure $\mu$, the tail symbols $\{X_{L+1},X_{L+2},\dots\}$ and $\{Y_{L+1},Y_{L+2},\dots\}$
		are independent i.i.d.\ uniform and are independent of the fixed prefixes $u,v$; moreover, the particular index $j$
		to be read at time $t$ is $\mathcal{F}_{t-1}$-measurable.
		Hence the next unread symbol that is read at time $t$ is uniform on $\A$ and independent of $\mathcal{F}_{t-1}$, and
		therefore $U[t]=Z[t]$ is uniform and independent of the past.
		In all cases, conditional on $\mathcal{F}_{t-1}$, the random variable $U[t]$ is uniform on $\A$ and independent of
		$\mathcal{F}_{t-1}$; thus $(U[1],\dots,U[n])$ are i.i.d.\ uniform and $U$ is uniform on $\A^n$.

		\smallskip
		\noindent\textbf{Claim 2:} Let $C:=\alocc_{w,r}(Z)$ and $C':=\alocc_{w,r}(U)$. Then $\abs{C-C'}\le 2L$.
		
		\smallskip
		\noindent
		Changing one symbol of a word can affect $\alocc_{w,r}$ by at most $1$, because each position lies in exactly one
		aligned length-$r$ block. The words $Z$ and $U$ differ only at positions in $T$, hence in at most $\abs{T}\le 2L$
		positions. Therefore $\abs{C-C'}\le 2L$.
		
		\smallskip
		Now, if $\abs{C-m/k^r}\ge \varepsilon_n m$, then by Claim~2,
		\[
		\abs{C'-m/k^r}\ \ge\ \varepsilon_n m - 2L\ \ge\ \frac{\varepsilon_n m}{2},
		\]
		using the hypothesis $4L\le \varepsilon_n m$.
		
		By Claim~1, the $m$ aligned blocks of $U$ are i.i.d.\ uniform over $\A^r$, so
		\[
		C' \sim \mathrm{Bin}\bigl(m,\;k^{-r}\bigr),
		\qquad \E[C']=m/k^r.
		\]
		Let $\mu_{C'}:=m/k^r$ and $t:=\varepsilon_n m/2$, so the relative deviation is
		\[
		\delta\ :=\ \frac{t}{\mu_{C'}}\ =\ \frac{\varepsilon_n m/2}{m/k^r}\ =\ \frac{\varepsilon_n k^r}{2}.
		\]
		By the hypothesis $\varepsilon_n\le k^{-r}$ we have $\delta\le 1/2\le 1$ and we may apply Lemma~\ref{lem:chernoff}:
		\[
		\Pr\bigl[\abs{C'-\mu_{C'}}\ge t\bigr]
		\ \le\ 2\exp\!\left(-\frac{\delta^2}{3}\mu_{C'}\right)
		\ =\ 2\exp\!\left(-\frac{1}{3}\cdot \frac{\varepsilon_n^2 k^{2r}}{4}\cdot \frac{m}{k^r}\right)
		\ =\ 2\exp\!\left(-\frac{\varepsilon_n^2}{12}\,m\,k^r\right).
		\]
		Since $\mu(E_S(n,r,w)^c\mid [u]\times[v])=\Pr[\abs{C-m/k^r}\ge \varepsilon_n m]$, the implication above yields the lemma.
	\end{proof}

	\subsection{Aggregating conditional failure probabilities via a rolling potential function}
	\label{subsec:potential}

	We now aggregate the conditional failure probabilities of all constraints at a fixed length $n$.
	
	\begin{definition}
		\label{def:Bn}
		For $n\ge 3$ and prefixes $u,v\in\A^L$, define
		\[
		B_n(u,v)\ :=\ \sum_{i=1}^{t_n}\ \sum_{r=1}^{\ell_n}\ \sum_{w\in\A^r}
		\mu\bigl(E_{S_i}(n,r,w)^c\mid [u]\times[v]\bigr).
		\]
	\end{definition}
	
	The next corollary shows that once the prefix length is at most $\sqrt{n}$, the total conditional bad mass at
	length $n$ is already extremely small.
	
	\begin{corollary}
		\label{cor:activation-small}
		There exists $n_0$ such that for all $n\ge n_0$, all $L\le \sqrt{n}$, and all $u,v\in\A^L$,
		\[
		B_n(u,v)\ \le\ \frac{1}{n^2}.
		\]
	\end{corollary}
	
	\begin{proof}
		Fix $n\ge n_0$ and $L\le \sqrt{n}$.
		For every $r\le \ell_n$ and all $n\ge 2r$ we have
		\[
		m_{n,r}=\lfloor n/r\rfloor \ \ge\ \frac{n}{r}-1.
		\]
		Therefore
		\[
		\varepsilon_n\, m_{n,r}
		\ \ge\
		2\sqrt{\frac{\log n\log_k n}{n}}\left(\frac{n}{r}-1\right)
		\ =\
		\frac{2\sqrt{n\log n\log_k n}}{r}\;-\;2\sqrt{\frac{\log n\log_k n}{n}}.
		\]
		Using $r\le \ell_n\le \frac13\log_k n$ gives
		\[
		\frac{2\sqrt{n\log n\log_k n}}{r}
		\ \ge\
		\frac{2\sqrt{n\log n\log_k n}}{(1/3)\log_k n}
		\ =\ 6\sqrt{\log k}\,\sqrt{n}.
		\]
		Hence, for all sufficiently large $n$ (increase $n_0$ if needed), the subtracted term
		$2\sqrt{\frac{\log n\log_k n}{n}}=o(\sqrt{n})$ is negligible and we obtain the uniform bound
		\[
		\varepsilon_n\, m_{n,r}\ \ge\ 4.9\sqrt{n}
		\]
		for all $r\le \ell_n$. In particular, for all $L\le \sqrt{n}$ we have
		\[
		4L\ \le\ 4\sqrt{n}\ \le\ \varepsilon_n\, m_{n,r}
		\]
		uniformly over all $r\le \ell_n$.
		Also, for $r\le \ell_n$ we have $k^r\le k^{\ell_n}\le n^{1/3}$, hence
		\[
		\varepsilon_n k^r \ \le\ 2\,n^{1/3}\sqrt{\frac{\log n\log_k n}{n}}
		\ =\ 2\,n^{-1/6}\sqrt{\log n\log_k n},
		\]
		which is $<1$ for all sufficiently large $n$; increase $n_0$ so that $\varepsilon_n \le k^{-r}$ holds for all
		$n\ge n_0$ and all $r\le \ell_n$.
		Thus Lemma~\ref{lem:cond-one} applies to every triple $(S_i,r,w)$.
		
		Therefore
		\[
		B_n(u,v)\ \le\ \sum_{i\le n}\ \sum_{r\le \ell_n}\ \sum_{w\in\A^r}
		2\exp\!\left(-\frac{\varepsilon_n^2}{12}\,m_{n,r}k^r\right).
		\]
		Using $m_{n,r}\ge n/(2r)$ and $\varepsilon_n^2 = 4\frac{\log n\log_k n}{n}$ gives
		\[
		\frac{\varepsilon_n^2}{12}\,m_{n,r}k^r
		\ \ge\ \frac{4}{12}\cdot \frac{\log n\log_k n}{n}\cdot \frac{n}{2r}\cdot k^r
		\ =\ \frac{1}{6}\cdot \frac{k^r}{r}\,\log n\,\log_k n.
		\]
		In particular, for every $r\ge 1$ we have $\frac{k^r}{r}\ge k$, so the exponent is at least
		\[
		\frac{k}{6}\,\log n\,\log_k n
		= \frac{k}{6\log k}\,(\log n)^2.
		\]
		Hence there is a constant $c>0$ with
		\[
		2\exp\!\left(-\frac{\varepsilon_n^2}{12}\,m_{n,r}k^r\right)
		\le 2e^{-c(\log n)^2}
		\le n^{-10}
		\]
		for all sufficiently large $n$ (increase $n_0$ so this holds).
		Now count terms:
		\[
		\sum_{i\le n}\sum_{r\le \ell_n}\sum_{w\in\A^r}1
		\le n\cdot \ell_n \cdot \sum_{r=1}^{\ell_n}k^r
		\le n\cdot \ell_n \cdot k^{\ell_n+1}
		\le n\cdot \ell_n \cdot k\cdot n^{1/3}
		= O(n^{4/3}\log n).
		\]
		Therefore
		\[
		B_n(u,v)\ \le\ O(n^{4/3}\log n)\cdot n^{-10}\ \le\ \frac{1}{n^2}
		\]
		for all large enough $n$.
	\end{proof}

	We will enforce all constraints at a sparse set of checkpoint lengths $N_j$. Our construction chooses the next symbols of $(x,y)$ greedily, at each prefix length $L$, so as to keep a suitable potential function small.
	To make this greedy choice work uniformly for every $L$, we use a rolling potential that only includes a checkpoint $N_j$
	once the current prefix length is at least its activation length $A_j=\sqrt{N_j}$.
	
	Define checkpoint lengths
	\[
	N_j := (j+m_0)^4,\qquad j\ge 1,
	\]
	where $m_0$ is a fixed constant chosen so that $\sum_{j\ge 1} 1/N_j^2 < 1/4$ and $N_1\ge n_0$
	(e.g.\ any $m_0$ large enough).
	
	Define activation lengths
	\[
	A_j := \sqrt{N_j} = (j+m_0)^2.
	\]
	
	\begin{definition}
		\label{def:potential}
		For each $L\ge 0$, define the active index set
		\[
		J(L):=\{\, j\ge 1:\ A_j\le L \le N_j\,\}.
		\]
		For prefixes $u,v\in\A^L$, define
		\[
		\Phi_L(u,v)\ :=\ \sum_{j\in J(L)} B_{N_j}(u,v).
		\]
	\end{definition}
	
	The next lemma gives the basic averaging identity for one-step extensions of the current prefixes; it will be used to justify the greedy choice of the next symbols.
	
	\begin{lemma}
		\label{lem:tower}
		Fix $n\ge 3$ and prefixes $u,v\in\A^L$. Then
		\[
		B_n(u,v)\ =\ \frac{1}{k^2}\sum_{a\in\A}\sum_{b\in\A} B_n(ua,vb).
		\]
	\end{lemma}
	\begin{proof}
		Each term $\mu(E^c\mid [u]\times[v])$ is a conditional probability under the product measure.
		Conditioning on the next symbols $(X_{L+1},Y_{L+1})$ yields the law of total probability: the conditional
		probability under $[u]\times[v]$ is the average of the conditional probabilities under the refined cylinders
		$[ua]\times[vb]$. Summing over all constraints preserves equality.
	\end{proof}
	The next lemma formalizes the averaging step that guarantees the existence of a one-symbol extension with small potential.
	
	\begin{lemma}
		\label{lem:one-step}
		For every $L\ge 0$ and every prefixes $u,v\in\A^L$,
		\[
		\min_{a,b\in\A}\Phi_{L+1}(ua,vb)\ \le\ \sum_{j\in J(L+1)} B_{N_j}(u,v).
		\]
	\end{lemma}
	\begin{proof}
		By Lemma~\ref{lem:tower}, for each fixed $j\in J(L+1)$, $B_{N_j}(u,v)$ equals the average of $B_{N_j}(ua,vb)$ over
		$(a,b)\in\A\times\A$. Summing over $j\in J(L+1)$ gives
		\[
		\sum_{j\in J(L+1)} B_{N_j}(u,v)\ =\ \frac{1}{k^2}\sum_{a,b\in\A}\Phi_{L+1}(ua,vb),
		\]
		so the minimum is at most the average.
	\end{proof}
	
	The next lemma packages the invariant needed for correctness: the rolling potential stays strictly below $1$.
	
	\begin{lemma}
		\label{lem:phi-bound}
		For every $L\ge 0$ there exist prefixes $u,v\in\A^L$ such that
		\[
		\Phi_L(u,v)\ <\ 1.
		\]
	\end{lemma}
	
	\begin{proof}
		We prove by induction on $L$ the stronger bound
		\[
		\Phi_L(u,v)\ \le\ \sum_{j:\ A_j\le L}\frac{1}{N_j^2}.
		\]
		Base $L=0$: $J(0)=\emptyset$, so $\Phi_0(\lambda,\lambda)=0$.
		
		Inductive step: let $(u,v)\in\A^L\times\A^L$ be the produced prefixes.
		For any $L$, the algorithm chooses $(a^\star,b^\star)$ minimizing $\Phi_{L+1}(ua,vb)$, hence by Lemma~\ref{lem:one-step}
		\[
		\Phi_{L+1}(ua^\star,vb^\star)\ \le\ \sum_{j\in J(L+1)} B_{N_j}(u,v).
		\]
		Split $J(L+1)$ into old-active indices and newly activated ones:
		\[
		J(L+1) = \bigl(J(L)\cap J(L+1)\bigr)\ \cup\ \{j:\ A_j=L+1\}.
		\]
		The first part is a subset of $J(L)$, so
		\[
		\sum_{j\in J(L)\cap J(L+1)} B_{N_j}(u,v)\ \le\ \Phi_L(u,v).
		\]
		For any newly activated $j$ with $A_j=L+1$, we have $L\le A_j\le \sqrt{N_j}$, so by Corollary~\ref{cor:activation-small}
		(applied with $n=N_j$ and prefix length $L\le \sqrt{N_j}$),
		\[
		B_{N_j}(u,v)\ \le\ \frac{1}{N_j^2}.
		\]
		Therefore
		\[
		\Phi_{L+1}(ua^\star,vb^\star)
		\le \Phi_L(u,v) + \sum_{j:\ A_j=L+1}\frac{1}{N_j^2}
		\le \sum_{j:\ A_j\le L+1}\frac{1}{N_j^2}.
		\]
		Finally, by choice of $m_0$, $\sum_{j\ge 1}1/N_j^2<1/4<1$, so $\Phi_L<1$ for all $L$ along the run.
	\end{proof}
	
	\subsection{The algorithm and correctness}
	\label{subsec:algorithm-correctness}
	
	We now give the greedy construction of the prefixes $(x{\upharpoonright}N,y{\upharpoonright}N)$.
	At each length $L$ we consider all one-step extensions of the current prefixes and choose the pair of next symbols
	that minimizes the next-step potential $\Phi_{L+1}$.  The preceding subsection shows that this choice keeps the rolling
	potential strictly below $1$ throughout the run; in particular, it will force every checkpoint family $F_{N_j}$ to hold.
	
	\begin{algorithm}[H]
		\caption{Compute $(x{\upharpoonright}N,y{\upharpoonright}N)$}
		\label{alg:compute}
		\begin{algorithmic}[1]
			\Require $N\in\N$
			\State $u\gets \lambda$, $v\gets \lambda$, $L\gets 0$
			\While{$L<N$}
			\State Compute $\Phi_{L+1}(ua,vb)$ for all $(a,b)\in\A\times\A$
			\State Choose $(a^\star,b^\star)$ minimizing $\Phi_{L+1}(ua,vb)$
			\State $u\gets ua^\star$, $v\gets vb^\star$, $L\gets L+1$
			\EndWhile
			\State \Return $(u,v)$
		\end{algorithmic}
	\end{algorithm}
	
	We next show that the invariant $\Phi_L<1$ forces every checkpoint constraint family $F_{N_j}$ to hold.
	
	To connect the potential bound to satisfaction of constraints at a fixed length $n$, we use the following integrality
	observation once the conditioning prefixes determine the first $n$ symbols on both tapes.
	
	\begin{lemma}
		\label{lem:integer}
		If $L\ge n$, then $B_n(u,v)$ is an integer (hence either $0$ or at least $1$).
	\end{lemma}
	\begin{proof}
		When $L\ge n$, the first $n$ symbols on each tape are fixed under $[u]\times[v]$.
		Thus each event $E_{S_i}(n,r,w)$ (and its complement) is determined, so its conditional probability is $0$ or $1$.
		Summing finitely many $0$--$1$ values gives an integer.
	\end{proof}
	
	The next proposition applies Lemma~\ref{lem:integer} at each checkpoint length $N_j$ to turn the strict inequality
	$\Phi_{N_j}<1$ into the exact satisfaction of all constraints in $F_{N_j}$.
	
	\begin{proposition}
		\label{prop:checkpoints}
		Let $(x,y)$ be the infinite output of Algorithm~\ref{alg:compute}. Then for every $j\ge 1$, $(x,y)\in F_{N_j}$.
	\end{proposition}
	
	\begin{proof}
		Fix $j$ and let $L:=N_j$. Then $j\in J(L)$ (because $A_j\le N_j\le N_j$), so by Lemma~\ref{lem:phi-bound},
		\[
		0\le B_{N_j}(x{\upharpoonright}L,y{\upharpoonright}L)\ \le\ \Phi_L(x{\upharpoonright}L,y{\upharpoonright}L)\ <\ 1.
		\]
		By Lemma~\ref{lem:integer} (with $n=L=N_j$), $B_{N_j}(x{\upharpoonright}L,y{\upharpoonright}L)$ is an integer, hence it
		must equal $0$. Since $B_{N_j}$ is a sum of non-negative conditional probabilities, the sum is $0$ iff each summand is $0$,
		i.e.\ every constraint in $F_{N_j}$ holds. Thus $(x,y)\in F_{N_j}$.
	\end{proof}
	
	We now derive normality for every shuffler output from the fact that all checkpoint constraints are met.
	
	\begin{proposition}
		\label{prop:all-shufflers-normal}
		For every shuffler $S$, the sequence $S(x,y)$ is normal. In particular, $x$ and $y$ are normal.
	\end{proposition}
	
	\begin{proof}
		Fix $S=S_i$. Since $t_{N_j}=N_j$, for all sufficiently large $j$ we have $i\le t_{N_j}$, hence
		$(x,y)\in F_{N_j}$ implies that for every $r\le \ell_{N_j}$ and every $w\in\A^r$,
		\[
		\left|\alocc_{w,r}(S(x,y){\upharpoonright}N_j) - \frac{m_{N_j,r}}{k^r}\right|
		< \varepsilon_{N_j}\, m_{N_j,r}.
		\]
		Dividing by $m_{N_j,r}$ gives aligned frequency error $<\varepsilon_{N_j}$ at the subsequence $N_j$.
		
		Now fix any block length $r\ge 1$ and any $w\in\A^r$. For large $j$, we have $r\le \ell_{N_j}$. Let $n$ be large and choose $j$ with $N_j\le n < N_{j+1}$. Let $z=S(x,y)$ and write
		$m_n=\lfloor n/r\rfloor$ and $m_j=\lfloor N_j/r\rfloor$.
		Since each additional aligned block can increase $\alocc_{w,r}$ by at most $1$, we have
		\[
		0\le \alocc_{w,r}(z{\upharpoonright}n)-\alocc_{w,r}(z{\upharpoonright}N_j)\le m_n-m_j.
		\]
		Hence
		\[
		\left|\frac{\alocc_{w,r}(z{\upharpoonright}n)}{m_n}-\frac{\alocc_{w,r}(z{\upharpoonright}N_j)}{m_n}\right|
		\le \frac{m_n-m_j}{m_n}.
		\]
		Also, since $0\le \alocc_{w,r}(z{\upharpoonright}N_j)\le m_j$, we have
		\[
		\left|\frac{\alocc_{w,r}(z{\upharpoonright}N_j)}{m_n}-\frac{\alocc_{w,r}(z{\upharpoonright}N_j)}{m_j}\right|
		=\alocc_{w,r}(z{\upharpoonright}N_j)\left|\frac{1}{m_n}-\frac{1}{m_j}\right|
		\le \frac{m_n-m_j}{m_n}.
		\]
		By the triangle inequality,
		\[
		\left|\frac{\alocc_{w,r}(z{\upharpoonright}n)}{m_n}-\frac{\alocc_{w,r}(z{\upharpoonright}N_j)}{m_j}\right|
		\le 2\cdot \frac{m_n-m_j}{m_n}.
		\]
		Moreover,
		\[
		\frac{m_n-m_j}{m_n}
		\le \frac{\lfloor (n-N_j)/r\rfloor+1}{\lfloor n/r\rfloor}
		\le \frac{(n-N_j)/r+1}{n/r-1}
		= \frac{n-N_j+r}{n-r}.
		\]
		For $n>2r$ this implies
		\[
		\frac{m_n-m_j}{m_n}
		\le 3\frac{n-N_j}{n}+3\frac{r}{n}
		\le 3\frac{N_{j+1}-N_j}{N_j}+3\frac{r}{N_j}.
		\]
		Finally, since $(x,y)\in F_{N_j}$, for all sufficiently large $j$ (so that $i\le t_{N_j}$ and $r\le \ell_{N_j}$) we have
		\[
		\left|\frac{\alocc_{w,r}(z{\upharpoonright}N_j)}{m_j}-\frac{1}{k^r}\right|<\varepsilon_{N_j}.
		\]
		Combining the last two displays gives
		\[
		\left|\frac{\alocc_{w,r}(z{\upharpoonright}n)}{m_n}-\frac{1}{k^r}\right|
		\le 2\cdot \frac{m_n-m_j}{m_n}+\varepsilon_{N_j}
		\le 6\frac{N_{j+1}-N_j}{N_j}+6\frac{r}{N_j}+\varepsilon_{N_j}.
		\]
		Since $N_j=(j+m_0)^4$, we have $(N_{j+1}-N_j)/N_j\to 0$ and $r/N_j\to 0$, and also $\varepsilon_{N_j}\to 0$.
		Therefore the aligned block frequency at length $n$ converges to $1/k^r$.
		
		By the equivalence between aligned-count normality and non-aligned count normality (see the remark in the preliminaries), $z$ is normal.
		
		Finally, taking the shuffler that always reads only tape $1$ (resp.\ tape $2$) gives $x$ and $y$ normal.
	\end{proof}

	\subsection{Polynomial-time implementation}
	\label{subsec:ptime-impl}
	
	We now verify that Algorithm~\ref{alg:compute} can be implemented within polynomial time (in bit-complexity) when the
	alphabet size $k=\abs{\A}$ is fixed.
	
	\begin{lemma}
		\label{prop:ptime-runtime}
		For fixed alphabet size $k=\abs{\A}$, Algorithm~\ref{alg:compute} runs in time $N^{O(1)}$ (bit-complexity) on input $N$.
	\end{lemma}
	
	\begin{proof}
		Fix $N\in\N$. At iteration $L$ (with current prefixes $u,v\in\A^L$), the algorithm evaluates
		$\Phi_{L+1}(ua,vb)$ for all $(a,b)\in\A\times\A$ and picks the minimizer.
		Since $k$ is fixed, it suffices to show that for each fixed $(a,b)$, $\Phi_{L+1}(ua,vb)$ can be computed in time
		polynomial in $N$, uniformly over all $L\le N$.
		
		Let $L'\!:=L+1$ and write $u'=ua$ and $v'=vb$. We compute
		\[
		\Phi_{L'}(u',v')=\sum_{j\in J(L')} B_{N_j}(u',v').
		\]
		If $j\in J(L')$, then $A_j\le L'$ implies $j+m_0\le \sqrt{L'}+m_0$, hence
		\[
		N_j=(j+m_0)^4 \ \le\ (\sqrt{L'}+m_0)^4 \ =\ O\bigl((L')^2\bigr)\ \le\ O(N^2).
		\]
		Also $A_j\le L'$ implies $j\le \sqrt{L'}$, so $\abs{J(L')}\le \lfloor \sqrt{L'}\rfloor = O(\sqrt{N})$.
		Thus, for each $L'\le N$, computing $\Phi_{L'}(u',v')$ reduces to computing
		$B_n(u',v')$ for $O(\sqrt{N})$ values $n$ with $n\le cN^2$ for a constant $c$.
		
		We next explain how to compute one term appearing in $B_n(u',v')$ by dynamic programming.
		Fix parameters
		\[
		n,\ i\le t_n=n,\ 1\le r\le \ell_n,\ w\in\A^r,\ u',v'\in\A^{L'}\ (L'\le n),
		\]
		and write $S_i=(Q_i,q_0,\delta,\tau)$. We compute
		\[
		p:=\mu\bigl(E_{S_i}(n,r,w)^c\mid [u']\times[v']\bigr).
		\]
		
		Let $m=\lfloor n/r\rfloor$.
		We run a forward DP over the production of the first $n$ output symbols.
		A DP state consists of
		\[
		(q,a,b,s,\sigma,c),
		\]
		where:
		\begin{itemize}
			\item $q\in Q_i$ is the current shuffler state,
			\item $a\in\{0,\dots,n\}$ and $b\in\{0,\dots,n\}$ are the numbers of symbols consumed from tapes 1 and 2
			(so the output time is $t=a+b$),
			\item $s\in\{0,\dots,r-1\}$ is the position inside the current aligned length-$r$ block ($s=t\bmod r$),
			\item $\sigma\in\{0,1\}$ indicates whether the current block built so far still matches the prefix of $w$
			(i.e.\ $\sigma=1$ iff the last $s$ symbols equal $w[1..s]$),
			\item $c\in\{0,\dots,m\}$ is the number of completed aligned blocks equal to $w$ among the first $\lfloor t/r\rfloor$
			completed blocks.
		\end{itemize}
		
		Initialize $\mathrm{DP}[q_0,0,0,0,1,0]=1$ and all other entries to $0$.
		For each state with $t=a+b<n$, we update as follows.
		Let $\ell=\tau(q)\in\{1,2\}$ be the chosen tape.
		If $\ell=1$, the next tape symbol is at position $a+1$; if $a+1\le L'$ it equals the fixed letter $u'[a+1]$,
		otherwise it is uniform over $\A$. Similarly for $\ell=2$ with $v'$ and $b+1$.
		
		Thus each transition from $(q,a,b,s,\sigma,c)$ branches over $\alpha\in\A$ with probability
		\[
		\Pr[\alpha]=
		\begin{cases}
			1, & \text{if the next position is fixed and equals }\alpha,\\
			0, & \text{if the next position is fixed and differs from }\alpha,\\
			1/k, & \text{if the next position is unfixed.}
		\end{cases}
		\]
		For each $\alpha$ with $\Pr[\alpha]>0$, set $q'=\delta(q,\alpha)$ and increment the appropriate head counter
		($a'=a+1,b'=b$ if $\ell=1$, else $a'=a,b'=b+1$).
		Update the within-block variables:
		\[
		\sigma'=
		\begin{cases}
			1, & \sigma=1 \text{ and } \alpha=w[s+1],\\
			0, & \text{otherwise,}
		\end{cases}
		\qquad
		s'=s+1.
		\]
		If $s'<r$ (block not finished), keep count $c' = c$ and store $(q',a',b',s',\sigma',c')$.
		If $s'=r$ (block finished), then set
		\[
		c'=
		\begin{cases}
			c+1, & \sigma'=1,\\
			c, & \sigma'=0,
		\end{cases}
		\qquad
		s'=0,\ \sigma'=1
		\]
		(start the next block fresh), and store $(q',a',b',0,1,c')$.
		All updates add probability mass accordingly.
		
		After processing all states with $a+b=n$, we obtain the distribution of the random variable
		$C=\alocc_{w,r}(S_i(X,Y){\upharpoonright}n)$ under the conditioning:
		\[
		\Pr[C=c]=\sum_{q\in Q_i}\sum_{a+b=n}\sum_{s,\sigma}\mathrm{DP}[q,a,b,s,\sigma,c].
		\]
		We then compute
		\[
		p=\Pr\Bigl[\bigl|C-m/k^r\bigr|\ge \varepsilon_n m\Bigr]
		\]
		by summing $\Pr[C=c]$ over all integers $c$ outside the allowed interval.
		
		The DP table has at most
		\[
		\abs{Q_i}\cdot (n+1)^2 \cdot r \cdot 2 \cdot (m+1) \ =\ O\!\bigl(\abs{Q_i}\,n^3\,r\bigr)
		\ =\ O\!\bigl(\abs{Q_i}\,n^3\log n\bigr)
		\]
		states, since $r\le \ell_n=O(\log n)$ and $m\le n$.
		Each state performs $O(k)=O(1)$ arithmetic updates.
		Probabilities are rationals with denominator at most $k^n$, so numerators/denominators have $O(n)$ bits;
		all additions/multiplications therefore cost polynomial time in $n$ (with fixed $k$).
		Hence this computation runs in time $n^{O(1)}\cdot \abs{Q_i}$.
		
		We now compute $B_n(u',v')$ itself. By definition,
		\[
		B_n(u',v')=\sum_{i=1}^{n}\ \sum_{r=1}^{\ell_n}\ \sum_{w\in\A^r}
		\mu\bigl(E_{S_i}(n,r,w)^c\mid [u']\times[v']\bigr).
		\]
		There are $n$ choices of $i$, $O(\log n)$ choices of $r$, and
		\[
		\sum_{r=1}^{\ell_n}\abs{\A^r}=\sum_{r=1}^{\ell_n}k^r \le k^{\ell_n+1} \le k\cdot n^{1/3}
		\]
		choices of $w$. Therefore $B_n(u',v')$ is a sum of at most $O(n^{4/3}\log n)$ terms.
		Computing each term as above and summing them yields a total time polynomial in $n$.
		
		Finally, for fixed $L'\le N$ and fixed $(a,b)$, we compute
		\[
		\Phi_{L'}(u',v')=\sum_{j\in J(L')} B_{N_j}(u',v').
		\]
		As shown at the start, $\abs{J(L')}=O(\sqrt{N})$ and every $N_j\le cN^2$.
		Each $B_{N_j}(u',v')$ is computable in time polynomial in $N_j$, hence polynomial in $N$,
		so $\Phi_{L'}(u',v')$ is computable in time polynomial in $N$.
		We do this for $k^2=O(1)$ candidate pairs $(a,b)$ at each of $N$ iterations, so the full run time is $N^{O(1)}$.
	\end{proof}

	Proposition~\ref{prop:all-shufflers-normal} establishes that the limit sequences $x,y$ are normal and that $S(x,y)$ is normal for every shuffler $S$, and Lemma~\ref{prop:ptime-runtime} shows that the construction runs in time $N^{O(1)}$ on input $N$. This completes the proof of Theorem~\ref{thm:main}. And finally due to Theorem \ref{thm:abc}, the pair $(x,y)$ is finite-state independent.

	\section{A computable finite-state independent word for a fixed computable normal input}
	\label{sec:computable-one-sided}
	
	Throughout this subsection, fix a finite alphabet $\A$ of size $k\ge 2$ and a computable normal word
	$x\in\A^\infty$. We construct a computable $y\in\A^\infty$ such that $y$ is normal and $S(x,y)$ is
	normal for every shuffler $S$. By the shuffler characterization theorem stated earlier in the paper,
	this implies that $x$ and $y$ are finite-state independent.
	
	\paragraph{Proof outline.}
	The construction has two ingredients.
	\begin{enumerate}
		\item We first show an almost-sure statement: for each fixed shuffler $S$, if $Y\sim\Unif(\A^\infty)$ then $S(x,Y)$ is normal with probability $1$.
		\item We then effectivize this: from a uniformly computable family of clopen sets of high measure (one for each finite collection of shuffler tests), we compute a single computable word $y$ that lies in all of them simultaneously.
	\end{enumerate}
	
	\subsection{Finite test sets and a computable extraction lemma}
	\label{subsubsec:comp-tests-extract}

	For parameters $m\ge 1$ and $\varepsilon>0$, write
	\[
	\mathrm{Test}(u;m,\varepsilon)\;:\Longleftrightarrow\;
	\forall w\in\A^m:\ \left|\frac{\occ(u,w)}{|u|}-k^{-m}\right|\le \varepsilon.
	\]
	For a shuffler $S$ and inputs $(x,y)$ define
	\[
	E_S(n;m,\varepsilon)
	:=\{(x,y)\in\A^\infty\times\A^\infty:\ \mathrm{Test}(S(x,y)[1..n];m,\varepsilon)\}.
	\]
	
	The next lemma records two basic facts needed for the effective intersection step: the relevant test sets are clopen, and their measures (and cylinder intersections) are uniformly computable when $x$ is computable.
	
	\begin{lemma}
		\label{lem:clopen-comp}
		For fixed $S,n,m,\varepsilon$, the set $E_S(n;m,\varepsilon)$ is clopen in $\A^\infty\times\A^\infty$.
		Moreover, for each fixed $x\in\A^\infty$, the slice
		\[
		E_S(n;m,\varepsilon)(x):=\{y:\ (x,y)\in E_S(n;m,\varepsilon)\}
		\]
		is clopen in $\A^\infty$.
		
		If $x$ is computable then for every cylinder $[v]\subseteq \A^\infty$ the quantity
		$\mu(E_S(n;m,\varepsilon)(x)\cap[v])$ is computable uniformly in $(S,n,m,\varepsilon,v)$.
		In particular, $\mu(E_S(n;m,\varepsilon)(x))$ is computable uniformly in $(S,n,m,\varepsilon)$.
	\end{lemma}
	
	\begin{proof}
		To determine membership in $E_S(n;m,\varepsilon)$ it suffices to know the output prefix
		$S(x,y)[1..n]$. During the first $n$ output steps, the shuffler reads at most $n$ symbols from each
		tape, hence $S(x,y)[1..n]$ is determined by $(x[1..n],y[1..n])$. Therefore $E_S(n;m,\varepsilon)$ is a
		finite union of rectangles $[u]\times[v]$ with $u,v\in\A^n$, and is clopen; each slice is a finite
		union of cylinders in $\A^\infty$.
		
		Assume $x$ is computable. Fix a cylinder $[v]$ with $|v|=r$.
		If $r\ge n$, then on $[v]$ the prefix $y[1..n]$ is fixed to be $v[1..n]$, hence either
		$[v]\subseteq E_S(n;m,\varepsilon)(x)$ or $[v]\cap E_S(n;m,\varepsilon)(x)=\emptyset$; in either case
		$\mu(E_S(n;m,\varepsilon)(x)\cap[v])\in\{0,k^{-r}\}$ is computable.
		
		If $r<n$, enumerate all $u\in\A^n$ extending $v$.
		For each such $u$, simulate $S$ for $n$ output steps on inputs $(x[1..n],u)$, decide whether the resulting
		output prefix passes $\mathrm{Test}(\cdot;m,\varepsilon)$, and count the number of passing $u$.
		Then $\mu(E_S(n;m,\varepsilon)(x)\cap[v])$ equals this count times $k^{-n}$, hence is a computable rational,
		uniformly in $(S,n,m,\varepsilon,v)$.
		\qedhere
	\end{proof}
	
	The next lemma is a standard effective intersection principle: from a uniformly computable sequence of high-measure clopen sets with rapidly summable error, one can produce a single computable point in their intersection.
	
	\begin{lemma}
		\label{lem:extract-comp}
		Let $G_1,G_2,\dots\subseteq \A^\infty$ be clopen sets such that for each $t$,
		\[
		1-\mu(G_t)\le 2^{-2t}.
		\]
		Assume that for every $L$ and cylinder $[v]$, the number
		\[
		\mu\!\left(\Bigl(\bigcap_{t\le L}G_t\Bigr)\cap[v]\right)
		\]
		is computable \emph{as an exact rational} uniformly in $(L,v)$ (hence comparisons with rational thresholds are decidable).
		Then there exists a computable $y\in \bigcap_{t\ge 1} G_t$.
	\end{lemma}
	
	\begin{proof}
		Write $\varepsilon_t:=2^{-2t}$ and let
		\[
		\mathrm{Tail}(L):=\sum_{t>L}\varepsilon_t.
		\]
		For $L\ge 0$ let $F_L:=\bigcap_{t\le L}G_t$ (with $F_0=\A^\infty$).
		
		We build prefixes $v_\ell\in\A^\ell$ such that
		\begin{equation}
			\label{eq:prefix-positive-mass}
			\mu\bigl([v_\ell]\cap \bigcap_{t\ge 1}G_t\bigr)>0
			\qquad(\ell\ge 0).
		\end{equation}
		This holds at $\ell=0$ because
		\[
		\mu\Bigl(\bigcap_{t\ge 1}G_t\Bigr)\ge 1-\sum_{t\ge 1}\mu(G_t^c)\ge 1-\sum_{t\ge 1}\varepsilon_t
		=1-\frac13=\frac23>0.
		\]
		
		Inductive step. Assume \ref{eq:prefix-positive-mass} holds for some $\ell$ and set
		\[
		\alpha:=\mu\bigl([v_\ell]\cap \bigcap_{t\ge 1}G_t\bigr)>0.
		\]
		Since $F_L\downarrow \bigcap_{t\ge 1}G_t$ as $L\to\infty$ and measures are continuous from above,
		\[
		\mu([v_\ell]\cap F_L)\downarrow \alpha \quad(L\to\infty).
		\]
		Also $\mathrm{Tail}(L)\to 0$. Hence there exists $L\ge 0$ such that
		\begin{equation}
			\label{eq:tail-thresh}
			\mu([v_\ell]\cap F_L)> k\,\mathrm{Tail}(L).
		\end{equation}
		By the hypothesis (exact rational computation), we can find the least such $L$ by brute force search.
		
		Fix this least $L$. Since
		\[
		\mu([v_\ell]\cap F_L)=\sum_{a\in\A}\mu([v_\ell a]\cap F_L),
		\]
		\eqref{eq:tail-thresh} implies that there exists at least one $a\in\A$ with
		\[
		\mu([v_\ell a]\cap F_L)>\mathrm{Tail}(L),
		\]
		otherwise the sum would be at most $k\,\mathrm{Tail}(L)$.
		Let $a$ be the least such symbol and set $v_{\ell+1}:=v_\ell a$.
		
		Now we show \ref{eq:prefix-positive-mass}. Using the union bound and $\mu([v_{\ell+1}]\cap G_t^c)\le \mu(G_t^c)\le\varepsilon_t$,
		\[
		\mu\Bigl([v_{\ell+1}]\cap \bigcap_{t\ge 1}G_t\Bigr)
		\ge
		\mu([v_{\ell+1}]\cap F_L)-\sum_{t>L}\mu([v_{\ell+1}]\cap G_t^c)
		\ge
		\mu([v_{\ell+1}]\cap F_L)-\mathrm{Tail}(L)>0,
		\]
		so \ref{eq:prefix-positive-mass} holds.
		
		Let $y:=\lim_\ell v_\ell\in\A^\infty$. Since the cylinders $[v_\ell]$ are nested, $y$ exists and lies in every $[v_\ell]$.
		Fix $t$. Since $G_t$ is clopen, membership in $G_t$ is decided by some prefix length $r$.
		Choose $\ell\ge r$. By \ref{eq:prefix-positive-mass}, $[v_\ell]$ intersects $\bigcap_{s\ge 1}G_s\subseteq G_t$, hence $[v_\ell]\cap G_t\neq\emptyset$.
		Because $|v_\ell|\ge r$ and $G_t$ is clopen, either $[v_\ell]\subseteq G_t$ or $[v_\ell]\cap G_t=\emptyset$.
		The latter is impossible, so $[v_\ell]\subseteq G_t$ and hence $y\in G_t$.
		As $t$ was arbitrary, $y\in\bigcap_{t\ge 1}G_t$.
		
		Computability: at stage $\ell$ we search effectively for the least $L$ satisfying \ref{eq:tail-thresh}, then compute
		$\mu([v_\ell a]\cap F_L)$ for each $a\in\A$ and choose the least $a$ with value $>\mathrm{Tail}(L)$.
		Thus $v_{\ell+1}$ is computable from $v_\ell$, and $y=\lim_\ell v_\ell$ is computable.
		\qedhere
	\end{proof}
	
	\subsection{Almost-sure normality for a fixed shuffler}
	\label{subsubsec:quenched-one-shuffler}
	
	Fix a shuffler $S=(Q,q_0,\delta,\tau)$ over $\A$. We prove:
	if $x$ is normal and $Y\sim\Unif(\A^\infty)$, then $S(x,Y)$ is normal almost surely.
	
	We decompose the run into segments according to how many symbols have been consumed from tape $X$.
	Let $Z:=S(x,Y)$. Let $i_t$ be the number of $x$-symbols consumed after producing $t$ output symbols.
	Fix an integer $N\ge 1$ and define stopping times
	\[
	T_0:=0,\qquad T_\ell:=\min\{t:\ i_t=\ell N\}\in\N\cup\{\infty\}.
	\]
	For each $\ell\ge 1$, define the $\ell$th segment by
	\[
	W_\ell :=
	\begin{cases}
		Z[T_{\ell-1}+1..T_\ell]\in\A^* & \text{if }T_\ell<\infty,\\[1mm]
		\lambda & \text{if }T_\ell=\infty,
	\end{cases}
	\]
	where $\lambda$ denotes the empty word (so $|W_\ell|=0$ when $T_\ell=\infty$).
	Note that $|W_\ell|\ge N$ whenever $T_\ell<\infty$, since in that case $W_\ell$ contains exactly $N$
	symbols copied from tape $X$.
	
	We will condition on the information revealed up to segment boundaries.
	Let $(\mathcal{G}_\ell)_{\ell\ge 0}$ be the natural filtration generated by the run up to time $T_\ell$,
	equivalently by the output prefix $Z[1..T_\ell]$ (and hence by the $Y$-symbols revealed up to time $T_\ell$).
	All conditional probabilities/expectations below are with respect to this filtration.
	
	We next record the block-frequency deviation statistic used within a segment.
	For a finite word $v$ and $m\ge 1$, write
	\[
	\occin(v,w):=\#\{1\le j\le |v|-m+1:\ v[j..j+m-1]=w\},
	\qquad
	\Delta_m(v):=\max_{w\in\A^m}\left|\frac{\occin(v,w)}{|v|}-k^{-m}\right|.
	\]
	
	We also isolate the part of the state space that is relevant on the event of infinitely many $X$-reads.
	Let $Q_X:=\tau^{-1}(1)$ and $Q_Y:=\tau^{-1}(2)$.
	Consider the directed graph on $Q$ with edges $q\to \delta(q,a)$ for $a\in\A$.
	Let $R\subseteq Q$ be the union of all \emph{sink strongly connected components} (i.e.\ strongly connected
	components with no outgoing edge) that intersect $Q_X$.
	
	For fixed $m$, set
	\[
	\mathcal{C}_R:=R\times \A^{\le m-1}.
	\]
	
	\begin{lemma}
		\label{lem:eventually-R}
		Define the segment-level hitting time
		\[
		\ell_R:=\min\{\ell\ge 1:\ \text{the shuffler state at time }T_{\ell-1}\text{ lies in }R\}\in\N\cup\{\infty\}.
		\]
		On the event that tape $X$ is read infinitely often (equivalently, $T_\ell<\infty$ for all $\ell$), we have
		$\ell_R<\infty$ and for all $\ell\ge \ell_R$ the segment-start state at time $T_{\ell-1}$ lies in $R$.
		Moreover, for each $L$, the event $\{\ell_R\le L\}$ is determined by the run up to time $T_{L-1}$
		(i.e.\ it is $\mathcal{G}_{L-1}$-measurable).
	\end{lemma}
	
	\begin{proof}
		For any infinite run of a finite-state automaton, the set of states visited infinitely often is contained
		in a sink strongly connected component (SCC) of the underlying transition graph.
		On the event that tape $X$ is read infinitely often, some state in $Q_X$ is visited infinitely often,
		so the sink SCC visited infinitely often intersects $Q_X$ and hence is contained in $R$.
		Therefore the run enters $R$ and never leaves it. Since $T_{\ell}\to\infty$ on this event, some segment start
		time $T_{\ell-1}$ occurs after the entry time into $R$, hence $\ell_R<\infty$ and all later segment-start states lie in $R$.
		
		The measurability claim holds because the shuffler state at time $T_{j}$ is determined by the run up to time $T_{j}$,
		so whether some $T_{j}$ with $j\le L-1$ has state in $R$ is determined by the run up to $T_{L-1}$. \qedhere
	\end{proof}
	
	Finally, we define the segment experiment from a boundary context.
For fixed $m$, let $\mathcal{C}_Q:=Q\times \A^{\le m-1}$ be the set of \emph{boundary contexts}, recording the current shuffler state together with the last $\le m-1$ output symbols.
 For $u\in\A^N$ and $\sigma\in\mathcal{C}_Q$, define $W(u,\sigma)$ as the random
	segment produced when:
	\begin{itemize}
		\item the next $N$ symbols read from tape $X$ are fixed to be $u$,
		\item the segment starts from boundary context $\sigma$,
		\item and every $Y$-symbol used during this segment is fresh i.i.d.\ uniform in $\A$.
	\end{itemize}
	(If the shuffler fails to consume $N$ symbols from tape $X$ under this experiment, we declare
	$W(u,\sigma)$ undefined.)
	
	We now identify those $X$-blocks whose induced segments typically have near-uniform internal block statistics, uniformly over the relevant boundary contexts.
	
	\begin{definition}
		\label{def:balanced-comp}
		Fix $m\ge 1$ and $\epsilon\in(0,1/10)$. A word $u\in\A^N$ is \emph{$(m,\epsilon)$-balanced} if for all
		$\sigma\in\mathcal{C}_R$,
		\[
		\Prob\bigl(W(u,\sigma)\ \text{is defined and}\ \Delta_m(W(u,\sigma))\le \epsilon\bigr)\ge 1-\epsilon,
		\]
		where the probability is over the fresh $Y$-symbols in the segment experiment.
	\end{definition}
	
	The next lemma shows that for sufficiently large segment length $N$, a uniformly random $X$-block is balanced with high probability.
	
	\begin{lemma}
		\label{lem:mostbalanced-comp}
		For each fixed $m$ and $\epsilon$, there exists $N_0=N_0(S,m,\epsilon)$ such that for all $N\ge N_0$,
		if $U\sim\Unif(\A^N)$ then
		\[
		\Prob\bigl(U\text{ is $(m,\epsilon)$-balanced}\bigr)\ge 1-\epsilon.
		\]
	\end{lemma}
	
	\begin{proof}
		Fix $\sigma\in\mathcal{C}_R$ and consider the fully random model $(X,Y)\sim\Unif(\A^\infty)^2$, started from boundary
		context $\sigma$. In this model, conditional on the past, the next output symbol is the next unused symbol
		from either tape, hence uniform and independent of the past; therefore the output process is i.i.d.\ uniform on $\A$.
		
		Let $U=X[1..N]$, and let $W_1$ be the first segment produced while consuming these $N$ symbols from tape $X$.
		Then $W_1$ has the same distribution as $W(U,\sigma)$ (when defined).
		Since $\sigma\in\mathcal{C}_R$ and $R$ is a union of sink strongly connected components intersecting $Q_X$,
		the state process stays inside the corresponding sink SCC; since this SCC is strongly connected, the induced finite Markov chain is
		irreducible, so every state in the SCC (in particular some state of $Q_X$) is visited infinitely often almost surely.
		In particular, $W_1$ is defined almost surely for every $N$.
		
		Because $|W_1|\ge N$, we have $|W_1|\to\infty$ as $N\to\infty$. Since the output is i.i.d.\ uniform,
		$\Delta_m(Z[1..L])\to 0$ almost surely as $L\to\infty$ (apply the strong law to each $w\in\A^m$ and union bound).
		Therefore $\Prob(\Delta_m(W_1)\le \epsilon)\to 1$ as $N\to\infty$.
		
		Now define $f_\sigma(u):=\Prob(W(u,\sigma)\ \text{defined and}\ \Delta_m(W(u,\sigma))\le\epsilon)$.
		Then
		\[
		\Exp_{U\sim\Unif(\A^N)}[f_\sigma(U)]=\Prob\bigl(\Delta_m(W_1)\le \epsilon\bigr).
		\]
		Thus for all sufficiently large $N$ we have $\Exp[f_\sigma(U)]\ge 1-\epsilon^2$. Markov’s inequality on
		$1-f_\sigma(U)$ yields $\Prob(f_\sigma(U)\le 1-\epsilon)\le\epsilon$.
		Finally, union bound over the finite set $\mathcal{C}_R$ completes the proof. \qedhere
	\end{proof}
	
	The next lemma quantifies how rare unbalanced blocks are, both in the uniform distribution on $\A^N$ and along the aligned $N$-block decomposition of a normal word.
	
	\begin{lemma}
		\label{lem:fewbad-comp}
		Let $B\subseteq\A^N$ be the set of $(m,\epsilon)$-unbalanced words. Then $|B|\le \epsilon k^N$.
		Moreover, if $x=u_1u_2\cdots$ is the aligned $N$-block decomposition of a normal word, then
		\[
		\limsup_{L\to\infty}\frac{\#\{1\le \ell\le L:\ u_\ell\in B\}}{L}\le 2\epsilon.
		\]
	\end{lemma}
	
	\begin{proof}
		The size bound is immediate from Lemma~\ref{lem:mostbalanced-comp}.
		For the density bound, normality implies that aligned $N$-blocks have limiting frequencies $k^{-N}$ for
		each $u\in\A^N$; summing over $u\in B$ gives limiting density $|B|/k^N\le\epsilon$, and the displayed
		bound allows slack. \qedhere
	\end{proof}
	
	We will use the standard bounded-difference Azuma--Hoeffding inequality: if $(D_j)_{j\ge 1}$ is a martingale difference sequence with
	$|D_j|\le 1$ almost surely, then for all $t>0$,
	\[
	\Prob\Bigl(\sum_{j=1}^L D_j\le -t\Bigr)\le \exp\!\left(-\frac{t^2}{2L}\right)
	\qquad\text{and}\qquad
	\Prob\Bigl(\sum_{j=1}^L D_j\ge t\Bigr)\le \exp\!\left(-\frac{t^2}{2L}\right).
	\]
	
	The next lemma shows that, along the segment decomposition, almost all segments have small internal $m$-block deviation (conditioned on tape $X$ being read infinitely often).
	
	\begin{lemma}
		\label{lem:manygood-comp-fixed}
		Fix $m,\epsilon$ and $N\ge N_0(S,m,\epsilon)$. Let $B_\ell$ be the event that $T_\ell<\infty$ and
		$\Delta_m(W_\ell)\le \epsilon$. Then on the event that $T_L<\infty$ for all $L$ (i.e.\ tape $X$ is read
		infinitely often), we have
		\[
		\Prob\left(\liminf_{L\to\infty}\frac{1}{L}\sum_{\ell=1}^L \mathbf{1}_{B_\ell}\ \ge\ 1-5\epsilon\right)=1.
		\]
	\end{lemma}
	
	\begin{proof}
		Write $x=u_1u_2\cdots$ in aligned $N$-blocks and call $\ell$ \emph{balanced} if $u_\ell$ is $(m,\epsilon)$-balanced.
		By Lemma~\ref{lem:fewbad-comp}, for all sufficiently large $L$, at least $(1-3\epsilon)L$ indices in $\{1,\dots,L\}$
		are balanced. Let $I_L:=\{\ell\le L:\ u_\ell\text{ is balanced}\}$, so $|I_L|\ge (1-3\epsilon)L$ for all large $L$.
		
		Work on the event $\mathcal I:=\{T_\ell<\infty\ \forall \ell\}$.
		For each $\ell$, let $p_\ell:=\Prob(B_\ell\mid \mathcal{G}_{\ell-1})$.
		Conditioned on $\mathcal{G}_{\ell-1}$, the next $N$ symbols to be consumed from tape $X$ are the fixed block $u_\ell$,
		and all as-yet-unread $Y$-symbols that may be used during segment $\ell$ are fresh i.i.d.\ uniform on $\A$;
		therefore, the conditional law of $W_\ell$ agrees with the segment experiment $W(u_\ell,\sigma_\ell)$ where
		$\sigma_\ell$ is the realized boundary context at time $T_{\ell-1}$.
		
		If $\ell\in I_L$ and $\ell\ge \ell_R$, then $\sigma_\ell\in\mathcal{C}_R$ and by balance of $u_\ell$ we have
		$p_\ell\ge 1-\epsilon$.
		
		Define, for each fixed $L$, the process
		\[
		M^{(L)}_j:=\sum_{\ell\le j,\ \ell\in I_L}\bigl(\mathbf{1}_{B_\ell}-p_\ell\bigr),
		\qquad j=0,1,\dots,L,
		\]
		with $M^{(L)}_0=0$.  Since $B_\ell$ is determined by the run up to time $T_\ell$, the indicator
		$\mathbf{1}_{B_\ell}$ is $\mathcal{G}_\ell$-measurable; by definition
		$p_\ell=\Prob(B_\ell\mid \mathcal{G}_{\ell-1})=\Exp[\mathbf{1}_{B_\ell}\mid \mathcal{G}_{\ell-1}]$.
		Therefore for each $\ell\le L$,
		\[
		\Exp\!\bigl[M^{(L)}_\ell-M^{(L)}_{\ell-1}\mid \mathcal{G}_{\ell-1}\bigr]
		=\mathbf{1}_{\{\ell\in I_L\}}\Exp\!\bigl[\mathbf{1}_{B_\ell}-p_\ell\mid \mathcal{G}_{\ell-1}\bigr]=0,
		\]
		so $(M^{(L)}_j)_{j=0}^L$ is a martingale with respect to $(\mathcal{G}_j)_{j\ge 0}$.
		Moreover, its increments are bounded: for $\ell\in I_L$,
		$\mathbf{1}_{B_\ell}-p_\ell\in[-1,1]$ since $\mathbf{1}_{B_\ell}\in\{0,1\}$ and $p_\ell\in[0,1]$,
		and for $\ell\notin I_L$ the increment is $0$.
		Finally, write $M_L:=M^{(L)}_L=\sum_{\ell\in I_L}(\mathbf{1}_{B_\ell}-p_\ell)$.
		The Azuma--Hoeffding inequality for martingales with bounded differences gives, for each $L$,
		\[
		\Prob\bigl(M_L\le -\epsilon|I_L|\bigr)\le \exp(-\epsilon^2|I_L|/2),
		\]
		and the right-hand side is summable in $L$ because $|I_L|\ge cL$ eventually.
		By Borel--Cantelli, almost surely for all sufficiently large $L$,
		\[
		\sum_{\ell\in I_L}\mathbf{1}_{B_\ell}\ \ge\ \sum_{\ell\in I_L}p_\ell\ -\ \epsilon|I_L|.
		\]
		
		On $\mathcal I$, we have $\ell_R<\infty$ by Lemma~\ref{lem:eventually-R}, hence for all large $L$,
		at most $\ell_R$ indices in $I_L$ are $<\ell_R$. Therefore,
		\[
		\sum_{\ell\in I_L}p_\ell\ \ge\ (1-\epsilon)\bigl(|I_L|-\ell_R\bigr).
		\]
		Combining,
		\[
		\sum_{\ell\in I_L}\mathbf{1}_{B_\ell}\ \ge\ (1-2\epsilon)|I_L|-(1-\epsilon)\ell_R
		\qquad\text{for all large }L \text{ on }\mathcal I.
		\]
		Since $\sum_{\ell=1}^L\mathbf{1}_{B_\ell}\ge \sum_{\ell\in I_L}\mathbf{1}_{B_\ell}$ and $|I_L|\ge (1-3\epsilon)L$ eventually,
		dividing by $L$ and letting $L\to\infty$ yields
		\[
		\liminf_{L\to\infty}\frac{1}{L}\sum_{\ell=1}^L \mathbf{1}_{B_\ell}\ \ge\ 1-5\epsilon
		\qquad\text{almost surely on }\mathcal I,
		\]
		as required. \qedhere
	\end{proof}
	
	We also need a uniform tail bound on segment lengths to rule out the possibility that a few exceptionally long segments dominate the statistics.
	
	\begin{lemma}
		\label{lem:exp-moment-length}
		There exist constants $\theta>0$ and $C\in(0,\infty)$, depending only on $S$, such that for every $N\ge 1$,
		every segment index $\ell$, and every history $\mathcal{G}_{\ell-1}$ on which segment $\ell$ starts in $R$,
		\[
		\Exp\bigl[e^{\theta |W_\ell|}\mid \mathcal{G}_{\ell-1}\bigr]\le e^{CN}.
		\]
	\end{lemma}
	
	\begin{proof}
		Since $R$ is finite and each sink strongly connected component in $R$ intersects $Q_X$, there exists
		$L_0$ such that from every $q\in R\cap Q_Y$ there is a word $a_1\cdots a_{L_0}$ whose induced state
		trajectory stays in $Q_Y$ until it first hits $Q_X$ at step $\le L_0$.
		Let $p:=k^{-L_0}$. Then, starting from any $q\in R\cap Q_Y$, the probability to hit $Q_X$ within $L_0$
		$Y$-steps is at least $p$, so the $Y$-waiting time until the next visit to $Q_X$ is stochastically dominated
		by $L_0\cdot\Geom(p)$ and hence has a uniform exponential moment for some $\theta>0$.
		
		Within segment $\ell$, exactly $N$ $X$-symbols are consumed. Between consecutive consumptions from $X$, the number
		of output steps spent reading from $Y$ has a conditional exponential moment bounded uniformly over segment-start
		histories with start state in $R$ by the preceding domination argument. Iterating $N$ times and applying the
		tower property yields $\Exp[e^{\theta|W_\ell|}\mid \mathcal{G}_{\ell-1}]\le e^{CN}$ for some $C$ depending only
		on $S$ and $\theta$. \qedhere
	\end{proof}
	
	As a direct consequence, we obtain an exponential tail bound for the segment length via Markov's inequality.
	
	\begin{lemma}
		\label{lem:length-tail}
		Assume segment $\ell$ starts in $R$. With $\theta,C$ as in Lemma~\ref{lem:exp-moment-length}, for every $t\ge 0$,
		\[
		\Prob\bigl(|W_\ell|>t\mid \mathcal{G}_{\ell-1}\bigr)\le \exp(CN-\theta t).
		\]
	\end{lemma}
	
	\begin{proof}
		Markov’s inequality and Lemma~\ref{lem:exp-moment-length} give
		\[
		\Prob(|W_\ell|>t\mid\mathcal{G}_{\ell-1})
		\le e^{-\theta t}\Exp[e^{\theta|W_\ell|}\mid\mathcal{G}_{\ell-1}]
		\le \exp(CN-\theta t).
		\]
		\qedhere
	\end{proof}
	
	We next show that segments longer than a fixed multiple of $N$ contribute only a vanishing fraction of the total output length.
	
	\begin{lemma}
		\label{lem:trunc-excess}
		Fix $N\ge 1$ and $D\ge 1$. On the event that tape $X$ is read infinitely often, almost surely
		\[
		\limsup_{L\to\infty}\frac{1}{LN}\sum_{\ell=1}^L |W_\ell|\mathbf{1}_{\{|W_\ell|>DN\}}
		\le \frac{8}{\theta N}\exp\!\left(\left(C-\frac{\theta D}{2}\right)N\right),
		\]
		where $\theta,C$ are as in Lemma~\ref{lem:exp-moment-length}.
	\end{lemma}
	
	\begin{proof}
		Work on the event $\mathcal I:=\{T_\ell<\infty\ \forall \ell\}$, so $\ell_R<\infty$ by Lemma~\ref{lem:eventually-R}.
		
		Fix $D\ge 1$. For $j\ge 0$ define the length class events
		\[
		A_{\ell,j}:=\{2^jDN<|W_\ell|\le 2^{j+1}DN\},
		\qquad
		N_{L,j}:=\sum_{\ell=1}^L \mathbf{1}_{\{\ell_R\le \ell\}}\mathbf{1}_{A_{\ell,j}}.
		\]
		By Lemma~\ref{lem:length-tail}, on the event $\{\ell_R\le \ell\}$ we have
		\[
		\Prob(A_{\ell,j}\mid \mathcal{G}_{\ell-1})
		\le
		\Prob(|W_\ell|>2^jDN\mid \mathcal{G}_{\ell-1})
		\le
		\exp\!\bigl(CN-\theta\cdot 2^jDN\bigr)
		=:p_j.
		\]
		Define the martingale (with respect to $(\mathcal{G}_\ell)_{\ell\ge 0}$)
		\[
		M_{L,j}:=\sum_{\ell=1}^L\left(\mathbf{1}_{\{\ell_R\le \ell\}}\mathbf{1}_{A_{\ell,j}}
		-\Exp[\mathbf{1}_{\{\ell_R\le \ell\}}\mathbf{1}_{A_{\ell,j}}\mid \mathcal{G}_{\ell-1}]\right),
		\]
		whose increments are bounded by $1$.
		Since $\Exp[\mathbf{1}_{\{\ell_R\le \ell\}}\mathbf{1}_{A_{\ell,j}}\mid \mathcal{G}_{\ell-1}]
		\le p_j\,\mathbf{1}_{\{\ell_R\le \ell\}}$, we have
		\[
		\sum_{\ell=1}^L \Exp[\mathbf{1}_{\{\ell_R\le \ell\}}\mathbf{1}_{A_{\ell,j}}\mid \mathcal{G}_{\ell-1}]
		\le p_j L.
		\]
		Hence the event $N_{L,j}\ge 2p_jL$ implies $M_{L,j}\ge p_jL$, and the Azuma--Hoeffding inequality yields
		\[
		\Prob(N_{L,j}\ge 2p_jL)\le \Prob(M_{L,j}\ge p_jL)\le \exp(-p_j^2L/2),
		\]
		which is summable in $L$. By Borel--Cantelli, almost surely for each fixed $j$, for all sufficiently large $L$,
		\[
		N_{L,j}\le 2p_jL,
		\]
		and hence $\limsup_{L\to\infty}N_{L,j}/L\le 2p_j$.
		
		Now,
		\[
		\sum_{\ell=1}^L |W_\ell|\mathbf{1}_{\{|W_\ell|>DN\}}
		=
		\sum_{\ell<\ell_R}|W_\ell|\mathbf{1}_{\{|W_\ell|>DN\}}
		+
		\sum_{\ell=\ell_R}^L |W_\ell|\mathbf{1}_{\{|W_\ell|>DN\}}.
		\]
		The first term is $O(1)$ on $\mathcal I$ and vanishes after dividing by $LN$.
		For the second term,
		\[
		\sum_{\ell=\ell_R}^L |W_\ell|\mathbf{1}_{\{|W_\ell|>DN\}}
		\le
		\sum_{j\ge 0} 2^{j+1}DN\cdot N_{L,j}.
		\]
		For each $J$,
		\[
		\limsup_{L\to\infty}\frac{1}{LN}\sum_{\ell=\ell_R}^L |W_\ell|\mathbf{1}_{\{|W_\ell|>DN\}}
		\le
		\sum_{j=0}^J 2^{j+2}D\cdot \limsup_{L\to\infty}\frac{N_{L,j}}{L},
		\]
		and letting $J\to\infty$ (monotone in $J$) gives
		\[
		\limsup_{L\to\infty}\frac{1}{LN}\sum_{\ell=\ell_R}^L |W_\ell|\mathbf{1}_{\{|W_\ell|>DN\}}
		\le
		\sum_{j\ge 0} 2^{j+2}D\cdot 2p_j
		=
		4D\,e^{CN}\sum_{j\ge 0} 2^j e^{-\theta DN 2^j}.
		\]
		We use the elementary bound: for all $a>0$,
		\[
		\sum_{j\ge 0}2^j e^{-a2^j}\le \frac{2}{a}e^{-a/2}.
		\]
		Indeed, for each $j\ge 0$ and all $u\in[2^{j-1},2^j]$ we have $e^{-au}\ge e^{-a2^j}$, hence
		\[
		\int_{2^{j-1}}^{2^j} e^{-au}\,du \ \ge\ 2^{j-1}e^{-a2^j}
		\quad\Rightarrow\quad
		2^j e^{-a2^j}\ \le\ 2\int_{2^{j-1}}^{2^j} e^{-au}\,du.
		\]
		Summing over $j\ge 0$ yields
		\[
		\sum_{j\ge 0}2^j e^{-a2^j}\le 2\int_{1/2}^{\infty} e^{-au}\,du = \frac{2}{a}e^{-a/2}.
		\]
		With $a:=\theta DN$, the right-hand side is at most
		\[
		\frac{8}{\theta N}\exp\!\left(\left(C-\frac{\theta D}{2}\right)N\right).
		\]
		\qedhere
	\end{proof}
	
	Finally, we record that the next segment becomes negligible compared to the total output produced so far.
	
	\begin{lemma}
		\label{lem:increment-vanish}
		On the event that tape $X$ is read infinitely often, we have
		\[
		\frac{|W_{L+1}|}{T_L}\to 0\qquad(L\to\infty)
		\]
		almost surely.
	\end{lemma}
	
	\begin{proof}
		Work on the event $\mathcal I:=\{T_\ell<\infty\ \forall \ell\}$, so $\ell_R<\infty$ by Lemma~\ref{lem:eventually-R}.
		Let
		\[
		b_L:=\frac{C}{\theta}N+\frac{2}{\theta}\log(L+2).
		\]
		By Lemma~\ref{lem:length-tail}, on the event $\{\ell_R\le L\}$ we have
		\[
		\Prob(|W_L|>b_L\mid \mathcal{G}_{L-1})\le \exp(CN-\theta b_L)\le (L+2)^{-2}.
		\]
		Hence
		\[
		\Prob(|W_L|>b_L\ \text{and}\ \ell_R\le L)\le (L+2)^{-2},
		\]
		and the right-hand side is summable. By Borel--Cantelli, almost surely only finitely many $L$ satisfy
		$|W_L|>b_L$ and $\ell_R\le L$. On $\mathcal I$ we have $\ell_R<\infty$, so for all sufficiently large $L$,
		$\ell_R\le L$ and thus $|W_L|\le b_L$.
		
		Also $T_L=\sum_{\ell\le L}|W_\ell|\ge LN$, hence for large $L$,
		\[
		\frac{|W_{L+1}|}{T_L}\le \frac{b_{L+1}}{LN}\to 0.
		\]
		\qedhere
	\end{proof}
	
	We now pass from segment-level control to prefix frequencies. For a concatenation $W_1\cdots W_L$ and any $w\in\A^m$, at most $(m-1)L$ occurrences of $w$ start in the last $(m-1)$ symbols of some $W_\ell$, hence cross a segment boundary. Apart from these boundary-crossing occurrences, the count of $w$ in $Z[1..T_L]=W_1\cdots W_L$ is the sum of the within-segment counts. The next lemma formalizes this bookkeeping: if most segments have small internal deviation and the total length contributed by very long segments is small, then the prefix $Z[1..T_L]$ has near-uniform $m$-block frequencies.
	
	\begin{lemma}
		\label{lem:prefixclose-comp-fixed}
		Fix $m\ge 1$, $\epsilon\in(0,1/10)$, and $N\ge N_0(S,m,\epsilon)$.
		Let $B_\ell$ be the event that $T_\ell<\infty$ and $\Delta_m(W_\ell)\le \epsilon$.
		Fix $D\ge 1$ and $\eta>0$. On the event that $T_L<\infty$ for all $L$, for all sufficiently large $L$,
		if
		\[
		\frac{1}{L}\sum_{\ell=1}^L \mathbf{1}_{B_\ell}\ \ge\ 1-6\epsilon
		\qquad\text{and}\qquad
		\frac{1}{LN}\sum_{\ell=1}^L |W_\ell|\mathbf{1}_{\{|W_\ell|>DN\}}\ \le\ \eta,
		\]
		then for every $w\in\A^m$,
		\[
		\left|\frac{\occ(Z[1..T_L],w)}{T_L}-k^{-m}\right|
		\le \epsilon\ +\ 6\epsilon D\ +\ \eta\ +\ \frac{m-1}{N}.
		\]
	\end{lemma}
	
	\begin{proof}
		Write $Z[1..T_L]=W_1\cdots W_L$. Boundary-crossing occurrences contribute at most $(m-1)L$, hence at most
		$(m-1)L/T_L\le (m-1)/N$ since $T_L=\sum_{\ell\le L}|W_\ell|\ge LN$.
		
		For internal occurrences, on $B_\ell$ we have $\occin(W_\ell,w)=k^{-m}|W_\ell|\pm \epsilon|W_\ell|$, so
		\[
		\sum_{\ell=1}^L \occin(W_\ell,w)
		=
		k^{-m}T_L\ \pm\ \epsilon T_L\ +\ \sum_{\ell\notin B}\occin(W_\ell,w),
		\]
		where $B:=\{\ell\le L:\ B_\ell\text{ holds}\}$. For $\ell\notin B$, use $\occin(W_\ell,w)\le |W_\ell|$ to get
		\[
		\left|\sum_{\ell=1}^L \occin(W_\ell,w)-k^{-m}T_L\right|
		\le \epsilon T_L\ +\ \sum_{\ell\notin B}|W_\ell|.
		\]
		Split the bad-length sum into short and long segments:
		\[
		\sum_{\ell\notin B}|W_\ell|
		\le
		\sum_{\ell\notin B,\,|W_\ell|\le DN}|W_\ell|
		\;+\;
		\sum_{\ell:\,|W_\ell|>DN}|W_\ell|.
		\]
		The first term is at most $DN\cdot \#\{\ell\notin B\}\le DN\cdot 6\epsilon L$. The second term is at most
		$\eta LN$ by hypothesis. Dividing by $T_L\ge LN$ yields
		\[
		\frac{1}{T_L}\sum_{\ell\notin B}|W_\ell|\le 6\epsilon D+\eta.
		\]
		Combine with the boundary term to obtain the stated inequality. \qedhere
	\end{proof}
	
	We can now combine the segment analysis to obtain the desired almost-sure normality statement for a fixed shuffler.
	
	\begin{theorem}
		\label{thm:oneS-general-comp}
		Let $S$ be any shuffler. If $x$ is normal and $Y\sim\Unif(\A^\infty)$, then $S(x,Y)$ is normal almost surely.
	\end{theorem}
	
	\begin{proof}
		If along the run the shuffler eventually reads only from tape $X$, then the output is a suffix of $x$ and hence normal.
		If it eventually reads only from tape $Y$, then the output is a suffix of $Y$ and hence normal almost surely.
		
		Let $\mathcal{I}$ be the remaining event that tape $X$ is read infinitely often (equivalently, $T_L<\infty$ for all $L$).
		It suffices to prove that on $\mathcal{I}$ the output $Z$ is normal almost surely.
		
		Fix $m\ge 1$ and $\delta>0$. Run the argument below with target accuracy $\delta/2$ and then extend to all $n$.
		
		Choose $D\ge 1$ so that $\theta D/2>C$, where $\theta,C$ are from Lemma~\ref{lem:exp-moment-length}.
		Choose $\epsilon\in(0,1/10)$ so that $\epsilon(1+6D)\le \delta/8$.
		Choose $N\ge N_0(S,m,\epsilon)$ so large that $(m-1)/N\le \delta/8$ and
		\[
		\frac{8}{\theta N}\exp\!\left(\left(C-\frac{\theta D}{2}\right)N\right)\le \delta/16.
		\]
		(Note that $(m-1)/N<1$ implies $N\ge m$.)
		
		By Lemma~\ref{lem:trunc-excess}, on $\mathcal{I}$ almost surely,
		\[
		\limsup_{L\to\infty}\frac{1}{LN}\sum_{\ell=1}^L |W_\ell|\mathbf{1}_{\{|W_\ell|>DN\}}
		\le \delta/16.
		\]
		Hence on $\mathcal{I}$ almost surely, for all sufficiently large $L$,
		\[
		\frac{1}{LN}\sum_{\ell=1}^L |W_\ell|\mathbf{1}_{\{|W_\ell|>DN\}}\le \delta/8.
		\]
		
		By Lemma~\ref{lem:manygood-comp-fixed}, on $\mathcal{I}$ almost surely,
		\[
		\liminf_{L\to\infty}\frac{1}{L}\sum_{\ell=1}^L \mathbf{1}_{B_\ell}\ge 1-5\epsilon.
		\]
		Hence on $\mathcal{I}$ almost surely, for all sufficiently large $L$,
		\[
		\frac{1}{L}\sum_{\ell=1}^L \mathbf{1}_{B_\ell}\ge 1-6\epsilon.
		\]
		
		Therefore, on $\mathcal{I}$ almost surely, for all sufficiently large $L$ the hypotheses of
		Lemma~\ref{lem:prefixclose-comp-fixed} hold with $\eta=\delta/8$, and hence
		\[
		\max_{w\in\A^m}\left|\frac{\occ(Z[1..T_L],w)}{T_L}-k^{-m}\right|\le \delta/2
		\]
		for all sufficiently large $L$.
		
		By Lemma~\ref{lem:increment-vanish}, $|W_{L+1}|/T_L\to 0$ on $\mathcal{I}$ almost surely. Fix such an $\omega$.
		Then for all sufficiently large $L$, for every $n\in[T_L,T_{L+1}]$ and every $w\in\A^m$, writing $r:=n-T_L\le |W_{L+1}|$,
		the number of new occurrences of $w$ created by extending from $T_L$ to $n$ is at most $r+(m-1)$, so
		\[
		|\occ(Z[1..n],w)-\occ(Z[1..T_L],w)|\le r+(m-1)\le |W_{L+1}|+ (m-1)\le 2|W_{L+1}|,
		\]
		using $|W_{L+1}|\ge N\ge m$.
		Therefore
		\[
		\left|\frac{\occ(Z[1..n],w)}{n}-\frac{\occ(Z[1..T_L],w)}{T_L}\right|
		\le
		\frac{2|W_{L+1}|}{T_L}+\frac{|W_{L+1}|}{T_L}
		=
		3\,\frac{|W_{L+1}|}{T_L}.
		\]
		For all sufficiently large $L$, the right-hand side is $\le \delta/2$, uniformly over $n\in[T_L,T_{L+1}]$ and $w\in\A^m$.
		Combining with the $\delta/2$ bound at time $T_L$ gives the desired $\delta$ bound for all sufficiently large $n$.
		
		Since $m$ and $\delta$ were arbitrary, $Z=S(x,Y)$ is normal almost surely. \qedhere
	\end{proof}
	
	\subsection{From almost-sure normality to a computable \texorpdfstring{$y$}{y}}
	\label{subsubsec:quenched-to-computable}
	
	We now compute a single computable $y$ that simultaneously forces $S(x,y)$ to be normal for every shuffler $S$.
	
	Fix a computable enumeration $S_1,S_2,\dots$ of all shufflers over $\A$.
	For $t\ge 1$ and $N\ge 2$, define the \emph{quadratic window slice}
	\[
	H_t(N)(x)
	:=
	\bigcap_{i=1}^t\ \bigcap_{\ell=1}^t\ \bigcap_{n=N}^{N^2}
	E_{S_i}\bigl(n;\ell,2^{-t}\bigr)(x)
	\ \subseteq\ \A^\infty.
	\]
	This is clopen (finite intersection of clopen sets). Since $x$ is computable, each measure
	$\mu(H_t(N)(x)\cap[v])$ is computable uniformly in $(t,N,v)$ by Lemma~\ref{lem:clopen-comp} and closure
	under finite intersections. The next lemma records that, for a fixed normal source $x$, the finite test family defining $H_t(N)(x)$ succeeds with probability tending to $1$ as the cutoff $N$ goes to infinity.
	
	\begin{lemma}
		\label{lem:sliceconv-comp}
		Fix $t\ge 1$ and a normal $x$. Then for $Y\sim\Unif(\A^\infty)$,
		\[
		\Prob\bigl(Y\in H_t(N)(x)\bigr)\to 1
		\qquad(N\to\infty).
		\]
	\end{lemma}
	
	\begin{proof}
		Fix $i\le t$ and $\ell\le t$. By Theorem~\ref{thm:oneS-general-comp}, $S_i(x,Y)$ is normal almost surely.
		For the fixed tolerance $2^{-t}>0$, almost surely there exists $N_{i,\ell}$ such that for all $n\ge N_{i,\ell}$,
		$\mathrm{Test}(S_i(x,Y)[1..n];\ell,2^{-t})$ holds. Let $N_*:=\max_{i\le t,\ \ell\le t}N_{i,\ell}$.
		Then almost surely, for all $N\ge N_*$, all tests in the definition of $H_t(N)(x)$ hold (since $[N,N^2]\subseteq[N_*,\infty)$),
		and hence $\mathbf{1}_{H_t(N)(x)}(Y)\to 1$ almost surely. By dominated convergence,
		$\Prob(Y\in H_t(N)(x))\to 1$. \qedhere
	\end{proof}
	
	We next choose concrete cutoffs $N_t$ so that each slice $H_t(N_t)(x)$ has very high measure, while keeping $N_t$ growing in a controlled way.
	
	\begin{lemma}
		\label{lem:choose-comp}
		Assume $x$ is computable and normal. There exists a computable increasing sequence
		$2\le N_1<N_2<N_3<\cdots$ such that for each $t\ge 1$,
		\[
		\mu\bigl(H_t(N_t)(x)\bigr)\ge 1-2^{-2t}
		\qquad\text{and}\qquad
		N_{t+1}\le N_t^2.
		\]
	\end{lemma}
	
	\begin{proof}
		We construct $(N_t)$ inductively by search with one-step lookahead.
		Here the searches are effective because each $H_t(N)(x)$ is clopen, so $\mu(H_t(N)(x))$ is a rational that can be computed exactly
		(by Lemma~\ref{lem:clopen-comp} and closure under finite intersections), hence the threshold tests are decidable.
		
		For $t=1$, search $N\ge 2$ until $\mu(H_1(N)(x))\ge 1-2^{-2}$ and there exists some $M$ with $N<M\le N^2$
		such that $\mu(H_2(M)(x))\ge 1-2^{-4}$ and there exists some $M'$ with $M<M'\le M^2$ such that
		$\mu(H_3(M')(x))\ge 1-2^{-6}$.
		Set $N_1:=N$ and $N_2$ to be the least such $M$.
		
		Now assume $t\ge 2$ and $N_t$ has been chosen.
		Search for the least $M$ with $N_t<M\le N_t^2$ such that
		\[
		\mu(H_{t+1}(M)(x))\ge 1-2^{-2(t+1)}
		\]
		and there exists some $M'$ with $M<M'\le M^2$ such that
		\[
		\mu(H_{t+2}(M')(x))\ge 1-2^{-2(t+2)}.
		\]
		Set $N_{t+1}:=M$.
		
		Existence follows from Lemma~\ref{lem:sliceconv-comp} (applied to the finitely many indices involved at each stage):
		since $\mu(H_s(N)(x))\to 1$ as $N\to\infty$ for each fixed $s$, one can take $M$ sufficiently large so that
		$\mu(H_{t+1}(M)(x))$ meets the stage threshold and also $M^2$ exceeds some stage-$(t+2)$ witness.
		By construction, $(N_t)$ is strictly increasing and satisfies $N_{t+1}\le N_t^2$ and the stated measure bounds. \qedhere
	\end{proof}
	
	Now, we prove the main result of this section.
	
	\begin{theorem}
		\label{thm:comp-independent}
		If $x\in\A^\infty$ is computable and normal, then there exists a computable $y\in\A^\infty$ such that:
		\begin{enumerate}[label=(\alph*)]
			\item $y$ is normal;
			\item for every shuffler $S$, the output $S(x,y)$ is normal.
		\end{enumerate}
		In particular, $x$ and $y$ are finite-state independent by the shuffler characterization theorem.
	\end{theorem}
	
	\begin{proof}
		Assume $x$ is computable and normal, and let $(N_t)_{t\ge 1}$ be the computable increasing sequence given by
		Lemma~\ref{lem:choose-comp}. Define $G_t:=H_t(N_t)(x)$. Then each $G_t$ is clopen with $\mu(G_t)\ge 1-2^{-2t}$, so
		\[
		1-\mu(G_t)\le 2^{-2t}
		\qquad\text{and}\qquad
		\sum_{t\ge 1}2^{-2t}=1/3<1.
		\]
		Apply Lemma~\ref{lem:extract-comp} to obtain a computable $y\in\bigcap_{t\ge 1}G_t$.
		
		Fix $i\ge 1$ and a block length $\ell\ge 1$, and let $z:=S_i(x,y)$.
		Let $\delta>0$ be arbitrary and choose $t\ge \max\{i,\ell\}$ with $2^{-t}\le \delta$.
		
		We claim that for every $n\ge N_t$,
		\[
		\mathrm{Test}\bigl(z[1..n];\ell,\delta\bigr)\ \text{holds}.
		\]
		Indeed, let $s\ge t$ be the maximal index such that $N_s\le n$ (well-defined since $N_s\to\infty$).
		Then $n< N_{s+1}\le N_s^2$ by Lemma~\ref{lem:choose-comp}, hence $n\in[N_s,N_s^2]$.
		Since $y\in G_s=H_s(N_s)(x)$, we have
		\[
		\mathrm{Test}\bigl(S_i(x,y)[1..n];\ell,2^{-s}\bigr)
		\]
		and since $2^{-s}\le 2^{-t}\le \delta$, this implies $\mathrm{Test}(z[1..n];\ell,\delta)$ as claimed.
		
		Because $\delta>0$ was arbitrary, it follows that for each fixed $\ell$, the $\ell$-block frequencies of $z$
		converge to $k^{-\ell}$, hence $z$ is normal. Since $i$ was arbitrary, this proves (b).
		
		For (a), let $S^{\mathrm{allY}}$ be the shuffler that always reads tape $Y$. Then $S^{\mathrm{allY}}(x,y)=y$,
		so applying (b) to $S^{\mathrm{allY}}$ shows that $y$ is normal. \qedhere
	\end{proof}

	\bibliography{main}
	\bibliographystyle{plain}

\end{document}